\documentclass[12pt,a4paper,onecolumn,accepted=2024-03-15]{quantumarticle}
\pdfoutput=1
\usepackage[utf8]{inputenc}
\usepackage[margin=2.6cm]{geometry}
\usepackage{amsmath,amsthm,amsfonts}
\usepackage{graphicx,titletoc}
\usepackage[titletoc]{appendix}
\usepackage[shortlabels]{enumitem}
\usepackage{hyperref}

\newtheorem{thm}{Theorem}[section]
\newtheorem{cor}[thm]{Corollary}
\newtheorem{lem}[thm]{Lemma}
\newtheorem{prop}[thm]{Proposition}
\theoremstyle{definition}
\newtheorem{defn}[thm]{Definition}
\newtheorem{rem}[thm]{Remark}

\newtheorem{thmA}{Theorem}

\DeclareMathOperator{\id}{id}
\DeclareMathOperator{\rk}{rk}
\DeclareMathOperator{\ran}{ran}
\DeclareMathOperator{\spa}{span}
\DeclareMathOperator{\tr}{tr}
\DeclareMathOperator{\sgn}{sgn}
\DeclareMathOperator{\Rep}{Rep}
\DeclareMathOperator{\mat}{mat}
\newcommand*{\mtxc}[1]{\operatorname{M}_{#1}(\C)}
\newcommand*{\mtxr}[1]{\operatorname{M}_{#1}(\R)}
\newcommand*{\glr}[1]{\operatorname{GL}_{#1}(\R)}
\newcommand*{\uc}[1]{\operatorname{U}_{#1}(\C)}
\newcommand{\bra}[1]{\mathinner{\langle #1|}}
\newcommand{\ket}[1]{\mathinner{|#1\rangle}}
\newcommand{\braket}[2]{\mathinner{\langle #1|#2\rangle}}
\def\C{\mathbb C}
\def\R{\mathbb R}
\def\N{\mathbb N}
\def\cA{\mathcal A}
\def\bA{\mathbf A}
\def\bB{\mathbf B}
\def\bC{\mathbf C}
\def\cH{\mathcal H}
\def\cK{\mathcal K}
\def\cM{\mathcal M}
\def\cN{\mathcal N}
\def\cS{\mathcal S}

\def\Cs{{\rm C}^*}
\def\wt{\widetilde}
\def\fP{\mathfrak{P}}
\def\fQ{\mathfrak{Q}}
\def\aux{{\rm aux}}
\def\ti{{\rm t}}
\def\ve{\varepsilon}

\title[Constant-sized self-tests]{Constant-sized self-tests for maximally entangled states and single local projective measurements}
\author{Jurij Vol\v{c}i\v{c}}
\address{Department of Mathematics, Drexel University, Pennsylvania}
\email{jurij.volcic@drexel.edu}
\thanks{Supported by the NSF grant DMS-1954709.}


\begin{document}
\maketitle

\begin{abstract}
Self-testing is a powerful certification of quantum systems relying on measured, classical statistics. This paper considers self-testing in bipartite Bell scenarios with small number of inputs and outputs, but with quantum states and measurements of arbitrarily large dimension. The contributions are twofold.
Firstly, it is shown that every maximally entangled state can be self-tested with four binary measurements per party. This result extends the earlier work of Man\v{c}inska-Prakash-Schafhauser (2021), which applies to maximally entangled states of odd dimensions only.
Secondly, it is shown that every single local binary projective measurement can be self-tested with five binary measurements per party. A similar statement holds for self-testing of local projective measurements with more than two outputs.
These results are enabled by the representation theory of quadruples of projections that add to a scalar multiple of the identity. Structure of irreducible representations, analysis of their spectral features and post-hoc self-testing are the primary methods for constructing the new self-tests with small number of inputs and outputs.
\end{abstract}

\section{Introduction}

Thanks to non-locality of quantum theory, unknown non-communicating quantum devices measuring an unknown shared entangled state can sometimes be identified based on classical statistic of their outputs. This phenomenon is called {\it self-testing}, and is the strongest form of device-independent certification of quantum systems. Self-testing was introduced in \cite{mayers}, and has been a heavily studied subject ever since; see \cite{supic} for a comprehensive review of major advances on this topic. The immense interest attracted by self-testing originates from its applications in device-independent quantum cryptography \cite{acin07,faleiro}, delegated quantum computation \cite{col19}, randomness generation \cite{miller16,bamps18}, entanglement detection \cite{bowles18}, and computational complexity \cite{fitzsimons,CECfalse}. For experimental developments, see \cite{loopholefree15,loopholefree23}.

This paper focuses on self-testing in bipartite Bell scenarios \cite{bell_review}, where two parties randomly perform measurements on a shared quantum state without communicating. 
From these measurements, joint probability distribution of inputs and outputs of both parties can be constructed as classical data describing the system. Suppose that each party can perform $N$ measurements, each of them with $K$ outcomes. Borrowing terminology from quantum games, we model this setup with bipartite quantum \emph{strategies}. Namely, an $N$-input $K$-output strategy $\cS$ of two parties (subsystems) A and B consists of a bipartite quantum state $\ket{\psi}$ in the tensor product of Hilbert spaces $\cH_A$ and $\cH_B$, a measurement $(\cM_{i,a})_{a=1}^K$ of positive operators on $\cH_A$ for each $i=1,\dots,N$, and a measurement $(\cN_{j,b})_{b=1}^K$ of positive operators on $\cH_B$ for each $j=1,\dots,N$.
The \emph{correlation} of $\cS$ is the array $p$ of probabilities given by the Born rule $p(a,b|i,j)=\bra{\psi}\cM_{i,a}\otimes \cN_{j,b}\ket{\psi}$, and is the classically observable data induced by $\cS$.
There are two trivial modifications of the strategy $\cS$ that do not affect its correlation: one is a unitary change of local bases, and the other is extending the state with an ancillary state on which the measurements act trivially. If any other strategy with correlation $p$ is obtained from $\cS$ using these trivial modifications, then we say that $\cS$ is \emph{self-tested} by $p$.
That is, the state and measurements in a self-tested strategy are essentially uniquely determined by the correlation.
The most renowned example of a self-tested strategy (with 2 inputs and 2 outputs) consists of maximally entangled qubits and two pairs of Pauli measurements, which give the maximal quantum violation of the famous CHSH inequality \cite{chsh69,tsirelson,mayers}.

The following is a fundamental self-testing problem: 

\vspace{1ex}\noindent ($\star$)\emph{
Which states and which measurements can be self-tested, i.e., appear in a strategy that is self-tested by its correlation? Furthermore, how complex is such a strategy, e.g., how many inputs and outputs per party are required?
}\vspace{1ex}

The breakthrough on ($\star$) for quantum states was achieved in \cite{CGS}, where the authors showed that every entangled bipartite state can be self-tested. The number of inputs in the provided self-tests \emph{grows with the local dimension $n$} of the quantum state under investigation, which makes these self-tests rather complicated in large dimensions. The existence result of \cite{CGS} was later not only extended to multipartite states in quantum networks \cite{networks23} and refined in one-sided device-independent scenarios \cite{sarkar23}, but also improved in terms of inputs and outputs needed to self-test certain states.
In \cite{sarkar21}, the authors show that an $n$-dimensional maximally entangled bipartite state can be self-tested using $2$ inputs and $n$ outputs. 
The paper \cite{fu} was the first to provide \emph{constant-sized} self-tests for some infinite families of maximally entangled states of \emph{even} dimension (but not constant-sized self-tests for all maximally entangled states of even dimension).
This result was complemented by \cite{mancinska}, where the authors establish that maximally entangled state of \emph{any odd} dimension can be self-tested using 4 inputs and 2 outputs.

In comparison with states, the progress on ($\star$) for measurements has been more constrained. All two-dimensional projective measurements have been self-tested \cite{yang13}, and likewise tensor products of Pauli measurements \cite{mckague,coladangelo}. Recently, it has been established that every projective measurement can be self-tested \cite{CMV}.
Actually, the self-tests derived in \cite{CMV} allow for arbitrary real ensembles of projective measurements to be self-tested simultaneously.
However, self-testing an $n$-dimensional projective measurement in this manner requires roughly $n^2$ inputs.

\subsection*{Contributions}

This paper provides self-tests for \emph{all maximally entangled states} and \emph{all single local projective measurements}, respectively, that are \emph{uniform} in number of both inputs and outputs.
The first main result concerns maximally entangled states.

\begin{thmA}[Corollary \ref{c:st4}]\label{ta}
Maximally entangled bipartite state of any local dimension $d$ can be self-tested using 4 inputs and 2 outputs.
\end{thmA}

The strategies of Theorem \ref{ta} are given in Definition \ref{d:4in}. Their construction and self-testing feature arises from the one-parametric family of universal C*-algebras $\cA_{2-\frac1n}$ generated by four
projections adding up to $2-\frac1n$ times the identity.
Remarkable results about representations of these algebras were established by Kruglyak-Rabanovich-Samo\u{\i}lenko using Coxeter functors between representation categories \cite{kruglyak}.
Their theory is essential in the proof of Theorem \ref{ta}.
Representations of C*-algebras of this type have already been leveraged in \cite{mancinska}. However, their work uses a different family of parameters ($2-\frac2n$ for odd $n$, instead of $2-\frac1n$ for natural $n$) that leads to simple C*-algebras, and maximally entangled states of odd dimensions only. On the other hand, exploiting algebras $\cA_{2-\frac1n}$ for self-testing purposes requires a more sophisticated analysis of their representations, but applies to all maximally entangled states.

The second main result of this paper provides constant-sized self-tests for single local projective measurements with 2 outputs, i.e., \emph{binary} projective measurements. Note that a local binary projective measurement $(P,I-P)$ is, 
up to unitary change of local basis, given by a real matrix,
 and determined by the dimension $n$ and the rank $r$ of the projection $P$.

\begin{thmA}[Corollary \ref{c:st5}]\label{tb}
A single local binary projective measurement of any dimension $n$ and rank $r$ appears in a 5-input 2-output strategy that is self-tested by its correlation.
\end{thmA}

See Definition \ref{d:5in} for the explicit strategies used in Theorem \ref{tb}. A generalization of Theorem \ref{tb} for local non-binary projective measurements is given in Corollary \ref{c:pvm}.
It is important to stress both the significance and the limitation of Theorem {\ref{tb}}. Given a single projective measurement, Theorem {\ref{tb}} provides a small self-testing strategy that contains this measurement. Note that up to a choice of coordinate system, a given projective measurement always admits a real matrix presentation.
However, Theorem {\ref{tb}} does not address self-testing of \textit{ensembles} of projective measurements; from this perspective, it is weaker than \mbox{\cite{CMV}}, which provides (large) self-tests for all real ensembles of projective measurements.
The strategies of Theorem \ref{tb} are obtained from the strategies of Theorem \ref{ta} by the principle of \emph{post-hoc self-testing} \cite{supic}. A broad sufficiency criterion for applicability of post-hoc self-testing was presented in \cite{CMV}. To apply this criterion in the proof of Theorem \ref{tb}, certain spectral aspects of representations of $\cA_{2-\frac1n}$ need to be resolved. Namely, we determine the spectrum of the sum of pairs of projections arising from representations of $\cA_{2-\frac1n}$.

While the derivation of the newly presented self-tests might seem rather abstract, the resulting correlations admit closed-form expressions, and the corresponding strategies can recursively constructed using basic tools from linear algebra (see Appendix \ref{app} for examples).

\subsection*{Reader's guide}
Section {\ref{s:prelim}} reviews the standard terminology and notation on quantum strategies and self-testing. Section {\ref{s:proj}} presents a construction of four $n\times n$ projections that add to $2-\frac1n$ times identity, and their basic properties; these projections are central to this paper, and provide local projective measurements for the new self-tested strategies. Section {\ref{s:spec}} establishes certain spectral results about these projections, which are critical for demonstrating self-testing in this paper. While this section provides the main new mathematical insight into what is required to establish the new self-testing results, a reader only interested in main statements may skip this section. Section {\ref{s:st}} presents the new self-tested strategies and their correlations. Section {\ref{s:obs}} addresses obstructions to constant-sized self-testing of arbitrary entangled states and pairs of projective measurements.
Lastly, Appendix {\ref{app}} explicitly constructs the distinguished projections appearing in self-tests for local dimensions up to 6.

\subsection*{Acknowledgments}
The author thanks Ken Dykema for inspiring conversations about self-testing, and Ricardo Gutierrez-Jauregui for sharing his expertise on experimental aspects of quantum theory.

\section{Preliminaries}\label{s:prelim}

This section introduces notation and terminology on quantum strategies and self-testing, following the conventions presented in \cite{mancinska}. For a comprehensive overview, see \cite{supic}.

Let $K\in\N$. A $K$-tuple of operators $(P_a)_{a=1}^K$ acting on a Hilbert space $\cH$ is a \emph{positive operator-valued measure ($K$-POVM)} if $P_a\succeq0$ and $\sum_{a=1}^K P_a=I$.
If all $P_a$ are projections, then $(P_a)_{a=1}^K$ is a \emph{projection-valued measure ($K$-PVM)}, or a \emph{projective measurement}. Note that, up to a unitary basis change, a PVM $(P_a)_{a=1}^K$ is uniquely determined by the ranks $\rk P_a$ for $a=1,\dots, K$. That is, every $K$-PVM with ranks of projections $r_1,\dots,r_K$ is unitarily equivalent to
$$\left(
I_{r_1}\oplus 0_{r_2+\cdots+r_K},
0_{r_1}\oplus I_{r_2}\oplus 0_{r_3+\cdots+r_K},\dots,
0_{r_1+\cdots+r_{K-1}}\oplus I_{r_K}
\right).$$
A $2$-POVM is also called a \emph{binary} measurement. Observe that a binary PVM is simply a pair $(P,I-P)$ where $P$ is a projection, and is determined by the dimension and the rank of $P$ up to a unitary basis change.

A \emph{(pure bipartite) state} $\ket{\psi}$ is a unit vector in $\cH_A\otimes\cH_B$, where $\cH_A,\cH_B$ are Hilbert spaces. We say that $\ket{\psi}$ has \emph{full Schmidt rank} if $P\otimes I\ket{\psi}=I\otimes Q\ket{\psi}=0$ for some projections $P,Q$ implies $P=0$ and $Q=0$. In this case, the Hilbert spaces $\cH_A$ and $\cH_B$ are isomorphic.
For $n\in\N$, the (canonical) \emph{maximally entangled state} of local dimension $n$ is $\ket{\phi_n}=\frac{1}{\sqrt{n}}\sum_{i=1}^n \ket{i}\!\ket{i}\in\C^n\otimes\C^n$. For $A,B\in\mtxc{n}$,
$$\bra{\phi_n}A\otimes B\ket{\phi_n}=\tau(AB^\ti)=\frac{1}{n}\tr(AB^\ti),$$
where $\tau$ denotes the normalized trace on $\mtxc{n}$.

Let $K_A,K_B,N_A,N_B\in\N$. An \emph{$(N_A,N_B)$-input $(K_A,K_B)$-output bipartite quantum strategy} $\cS$ is a triple
$$\cS=\left(
\ket{\psi};\cM_1,\dots,\cM_{N_A};\cN_1,\dots,\cN_{N_B}
\right)$$
where $\cM_i$ are $K_A$-POVMs on a finite-dimensional Hilbert space $\cH_A$, $\cN_j$ are $K_B$-POVMs on a finite-dimensional Hilbert space $\cH_B$, and $\ket{\psi}\in\cH_A\otimes\cH_B$ is a state. When $K=K_A=K_B$ and $N=N_A=N_B$, we simply say that $\cS$ is a \emph{$N$-input $K$-output bipartite strategy}.
The \emph{correlation} of $\cS$ is the $N_A\times N_B\times K_A\times K_B$ array $p$ with entries
\begin{align*}
p(a,b|i,j) = \bra{\psi}\cM_{i,a}\otimes \cN_{j,b}\ket{\psi}
\qquad\qquad& 1\le a\le K_A,\ 1\le b\le K_B,\\
\qquad\qquad&1\le i\le N_A,\ 1\le j\le N_B.
\end{align*}
Since $\cS$ in particular models non-communication between parties,
the correlation $p$ is \emph{non-signalling}, meaning that
$p(a|i):=\sum_{b=1}^{K_B}p(a,b|i,j)$ and $p(b|j):=\sum_{a=1}^{K_A}p(a,b|i,j)$ are well-defined (the first sum is independent of $j$ and the second sum is independent of $i$).
A correlation $p$ is called \emph{synchronous} if $K_A=K_B$, $N_A=N_B$ and $p(a,b|i,i)=0$ for all $i$ and $a\neq b$. 

Let $\cS$ and $\wt\cS$ be $(N_A,N_B)$-input $(K_A,K_B)$-output strategies. Then $\wt\cS$ is a \emph{local dilation} if $\cS$ there exist finite-dimensional Hilbert spaces $\cK_A,\cK_B$, a state $\ket{\aux}\in\cK_A\otimes \cK_B$ and isometries $U_A:\cH_A\to\wt\cH_A\otimes \cK_A$ and $U_B:\cH_B\to\wt\cH_B\otimes \cK_B$ such that
\begin{equation}\label{e:locdil}
(U_A\otimes U_B)(\cM_{i,a}\otimes \cN_{j,b})\ket{\psi}=(\wt\cM_{i,a}\otimes \wt\cN_{j,b})\ket{\wt\psi}\otimes \ket{\aux}
\end{equation}
for all $a,b,i,j$.
There is a slight abuse of notation in \eqref{e:locdil}; namely, we identify $$(\wt\cH_A\otimes\cK_A)\otimes (\wt\cH_B\otimes\cK_B)\equiv
(\wt\cH_A\otimes\wt\cH_B)\otimes (\cK_A\otimes\cK_B).$$
Note that if $\wt\cS$ is a local dilation of $\cS$, then the correlations of $\cS$ and $\wt\cS$ coincide.
Finally, we say that a strategy $\wt\cS$ is \emph{self-tested} by its correlation if it is a local dilation of any other strategy with the same correlation.

\section{Quadruples of projections adding to a scalar multiple of the identity}\label{s:proj}

In \cite{kruglyak}, the authors derive several profound results on tuples of projections that add to a scalar multiple of the identity operator. This is achieved by studying certain functors between categories of their representations, which are also the cornerstone of this paper.
For our purposes, we focus on projections $P_1,P_2,P_3,P_4$ that add to $(2-\frac1n)I$, where $n$ is a natural number.
First we adopt the language of representations of C*-algebras, at least to the extent required in this paper.
Then we review the construction of the aforementioned functors from \cite[Section 1.2]{kruglyak}. Finally, we refine a part of \cite[Proposition 3]{kruglyak} to obtain further properties about the projections $P_i$ as above (Proposition \ref{p:ukr}). 

For $\alpha\in\R$ define the universal C*-algebra
$$\cA_\alpha=\Cs\left\langle x_1,x_2,x_3,x_4\colon x_i=x_i^*=x_i^2,\ x_1+x_2+x_3+x_4=\alpha\right\rangle,$$
and let $\Rep_\alpha$ denote the category of representations of $\cA_\alpha$. That is, objects of $\Rep_\alpha$ are representations of $\cA_\alpha$ on Hilbert spaces, and morphisms of $\Rep_\alpha$ are equivariant maps, i.e., bounded linear operators between Hilbert spaces that intertwine the actions of representations.
For a comprehensive source on C*-algebras and their representations, see \mbox{\cite{blackadar}}. While the above terminology offers a suitable mathematical framework for the technical steps in the proofs of this paper, let us extract the main meaning behind it, sufficient for comprehending the proofs. 
Without addressing precisely what a universal C*-algebra is, we can still say what its representations are. A representation $\pi$ of $\cA_\alpha$ is a quadruple of projections $X_1,X_2,X_3,X_4$ on a Hilbert space $\cH$ that satisfy $X_1+X_2+X_3+X_4=\alpha I$. Thus $\Rep_\alpha$ is foremost a collection of such quadruples; one could think of $\cA_\alpha$ as their abstract model. For a $\pi\in\Rep_\alpha$ as above we write $\pi(x_i)=X_i$, and we assign to it a $6$-tuple of numbers $[\pi]=(\alpha;n;d_1,d_2,d_3,d_4)$
where $n=\dim \cH$ and $d_i=\rk \pi(x_i)$, the dimension of the range of $X_i$ (if $\cH$ is infinite-dimensional, then $n=\infty$; likewise, $d_i$ can be infinite).

Note that representations may be related to each other in several ways. 
For example, let $\pi\in \cA_\alpha$ is given by projections $X_1,\dots,X_4$ on a Hilbert space $\cH$ and $\rho\in \cA_\alpha$ is given by projections $Y_1,\dots,Y_4$ on a Hilbert space $\cK$. Then the projections $X_1\oplus Y_1,\dots,X_4\oplus Y_4$ act on $\cH\oplus \cK$ and add to $\alpha$ times identity, so they determine representation of $\cA_\alpha$, called the \textit{direct sum} of $\pi$ and $\rho$.
Next, we say that $\pi$ and $\rho$ are \textit{unitarily equivalent} if there is a unitary (that is, an isometric invertible linear map) $U:\cH\to \cK$ such that $Y_i=UX_iU^*$ for $i=1,\dots,4$. Finally, we say that $\pi\in\Rep_\alpha$ is \textit{irreducible} if it is not unitarily equivalent to a direct sum of representations. Irreducible representations can be viewed as the building blocks of $\Rep_\alpha$; namely, every representation is unitarily equivalent to a (possibly infinite) direct sum of irreducible representations.
Without going into technical details, viewing $\Rep_\alpha$ as a category instead of merely a set encapsulates these relations between representations (e.g., that some of them are unitarily equivalent, some are direct sums of others, and some are irreducible). 

In this paper,  representations of $\cA_\alpha$ (for certain choices of $\alpha$) give rise to the projective measurements in self-tested strategies presented in Section {\ref{s:st}}. To establish the self-testing property, it is imperative to have a good handle on $\Rep_\alpha$ (concretely, on the irreducible representations within). 
This is straightforward for $\alpha=0$ and $\alpha=1$. Indeed, the only quadruples of projections adding to 0 are tuples of zero operators; these are all direct sums of the trivial representation $\tau$ given by $\tau(x_j)=0$ acting on the one-dimensional Hilbert space. Hence $\Rep_0$ contains a unique irreducible representation. On the other hand, quadruples of projections adding to 1 are necessarily diagonalizable, and thus unitarily equivalent to direct sums of $(1,0,0,0)$, $(0,1,0,0)$, $(0,0,1,0)$, $(0,0,0,1)$ acting on the one-dimensional Hilbert space. Thus $\Rep_1$ contains exactly four unitarily non-equivalent irreducible representations.
For general $\alpha$, representations of $\cA_\alpha$ are not yet well-understood; however, the aim of the next subsection is to leverage the knowledge of the very simple $\Rep_1$ to study $\Rep_\alpha$ for certain values of $\alpha$.

\subsection{Functors between representation categories}

In this subsection we define two functors $T=T_\alpha:\Rep_\alpha\to \Rep_{4-\alpha}$ (linear reflection) and $S=S_\alpha:\Rep_\alpha\to \Rep_{\frac{\alpha}{\alpha-1}}$ (hyperbolic reflection). The subscripts are omitted when clear from the context.
Before defining $T$ and $S$, let us mention what a reader should imagine under this terminology. A functor from $\Rep_\alpha$ to $\Rep_\beta$ is primarily a mapping, that takes each quadruple of projections adding to $\alpha$ times identity to a quadruple of projections adding to $\beta$ times identity. However, being a functor means that this mapping has to respect the additional structure of the categories $\Rep_\alpha$ and $\Rep_\beta$; in particular, it needs to preserve direct sums, and map unitarily equivalent representations to unitarily equivalent representations. Technically, one encapsulates this by saying that a functor consists of a map between objects of categories and a (well-behaved) map between morphisms of categories.

$(T)$: Given a representation $\pi$ of $\cA_\alpha$ let $T(\pi)$ be the representation of $\cA_{4-\alpha}$ determined by $T(\pi)(x_i) := I-\pi(x_i)$. Note that $T$ commutes with equivariant maps between representations, so it extends to a functor $T:\Rep_\alpha\to \Rep_{4-\alpha}$.
If $[\pi]=(\alpha;n;d_i)$ then $[T(\pi)]=(4-\alpha;n;n-d_i)$.

($S$): Suppose $\alpha\notin\{0,1\}$, and let $\pi$ be a representation of $\cA_\alpha$ on $\cH$. 
Denote $\widehat{\cH}=\bigoplus_i \ran \pi(x_i)$. Let $w_i:\ran \pi(x_i)\to \widehat{\cH}$ be the canonical injections, 
and let $u_i:\ran \pi(x_i)\to \cH$ be inclusions. Then
$$u=\frac{1}{\sqrt{\alpha}}\begin{pmatrix}
u_1^* \\ \vdots \\ u_4^*
\end{pmatrix}:\cH\to\widehat{\cH}$$
is an isometry by definition of the algebra $\cA_\alpha$. Let $\cK=\ran(I-uu^*)$, with inclusion $v:\cK\to\widehat{\cH}$. Note that $\dim\cK=\dim\widehat\cH-\dim\cH$. Define
$$S(\pi)(x_i):=\frac{\alpha}{\alpha-1}
v^*w_iw_i^*v.$$
Then
\begin{align*}
\left(S(\pi)(x_i)\right)^2
&=\frac{\alpha^2}{(\alpha-1)^2}v^*w_iw_i^*vv^*w_iw_i^*v 
=\frac{\alpha^2}{(\alpha-1)^2}v^*w_iw_i^*(I-uu^*)w_iw_i^*v \\
&=\frac{\alpha^2}{(\alpha-1)^2}v^*w_i\left(I-\frac{1}{\alpha}u_i^*u_i\right)w_i^*v 
=\frac{\alpha^2}{(\alpha-1)^2}\left(1-\frac{1}{\alpha}\right)v^*w_iw_i^*v\\
&=S(\pi)(x_i)
\end{align*}
and
$$\sum_{i=1}^4S(\pi)(x_i)
=\sum_{i=1}^4\frac{\alpha}{\alpha-1}v^*w_iw_i^*v 
=\frac{\alpha}{\alpha-1}v^*\left(\sum_{i=1}^4 w_iw_i^*\right) v 
=\frac{\alpha}{\alpha-1}v^*v=\frac{\alpha}{\alpha-1}I. 
$$
Therefore $S(\pi)(x_1),\dots,S(\pi)(x_4)$ are projections that give rise to a representation $S(\pi)$ of $\cA_{\frac{\alpha}{\alpha-1}}$ on $\cK$.
As described in \cite[Section 1.2]{kruglyak}, one can also extend $S$ to equivariant maps, resulting in a functor $S:\Rep_\alpha\to \Rep_{\frac{\alpha}{\alpha-1}}$.
If $[\pi]=(\alpha;n;d_i)$ then $[S(\pi)]=(\frac{\alpha}{\alpha-1};\sum_id_i-n;d_i)$. 

\subsection{Distinguished quadruples of projections}

For $\alpha\in(0,3)$, the (Coxeter) functor
$$\Phi^+=S\circ T=S_{4-\alpha}\circ T_\alpha: \Rep_\alpha\to \Rep_{1+\frac{1}{3-\alpha}}$$
define an equivalence of categories (with inverse $T\circ S$) by \cite[Theorem 2]{kruglyak}.
In particular, $\Phi^+$ is a bijection between representations of $\cA_\alpha$ and $\cA_{1+\frac{1}{3-\alpha}}$, which maps irreducible ones to irreducible ones.
If $[\pi]=(\alpha,n,d_1,\dots,d_4)$ then $[\Phi^+(\pi)]=(1+\frac{1}{3-\alpha};3n-\sum_id_i;n-d_i)$. 
The functor $\Phi^+$ plays an implicit yet crucial role in \cite[Proposition 3]{kruglyak} that describes the category $\Rep_{2-\frac{1}{n}}$. For the sake of completeness, we provide the proof of the part of \cite[Proposition 3]{kruglyak}, and refine it to extract the additional information needed in this paper.
Given a real number $\beta$ let $\lfloor \beta\rfloor$ denote the largest integer that is not larger than $\beta$.

The main statement of this section shows that starting with the easily-understood $\Rep_1$ and then repeatedly applying the functor $\Phi^+$, one obtains a good grasp on $\Rep_{2-\frac1n}$ for every $n\in\N$.

\begin{prop}[{\cite[Proposition 3(c)]{kruglyak}}]\label{p:ukr}
Let $n\in\N$. The C*-algebra $\cA_{2-\frac{1}{n}}$ has precisely four unitarily non-equivalent irreducible representations. \\
More concretely, there are projections $\fP^{(n)}_1,\dots,\fP^{(n)}_4\in\mtxr{n}$ with $\rk \fP^{(n)}_1=\lfloor\frac{n}{2}\rfloor-(-1)^n$ and $\rk \fP^{(n)}_i=\lfloor\frac{n}{2}\rfloor$ for $i=2,3,4$, such that given an irreducible representation of $\cA_{2-\frac{1}{n}}$, the quadruple $(\pi(x_1),\dots,\pi(x_4))$ is unitarily equivalent to one of the
\begin{align*}
&(\fP^{(n)}_1,\fP^{(n)}_2,\fP^{(n)}_3,\fP^{(n)}_4),\quad 
(\fP^{(n)}_4,\fP^{(n)}_1,\fP^{(n)}_2,\fP^{(n)}_3),\\ 
&(\fP^{(n)}_3,\fP^{(n)}_4,\fP^{(n)}_1,\fP^{(n)}_2),\quad 
(\fP^{(n)}_2,\fP^{(n)}_3,\fP^{(n)}_4,\fP^{(n)}_1).
\end{align*}
\end{prop}

\begin{proof}
We prove the statement by induction on $n$.
If $n=1$, then $\fP^{(1)}_1=1$ and $\fP^{(1)}_i=0$ for $i=2,3,4$ are the desired $1\times 1$ projections, giving rise to a representation $\cA_1\to\C$. Now suppose projections $\fP^{(n)}_i\in\mtxr{n}$ possess the desired properties.
Then they define an irreducible representation of $\cA_{2-\frac{1}{n}}$ given by $\pi(x_i)=\fP^{(n)}_i$, and the other three irreducible representations up to unitary equivalence are obtained by cyclically permuting the generators. Now let $\fP^{(n+1)}_i:=\Phi^+(\pi)(x_i)$. Since $\Phi^+:\Rep_{2-\frac{1}{n}}\to\Rep_{2-\frac{1}{n+1}}$ is an equivalence of categories, $\Phi^+(\pi)$ is an irreducible representation of $\cA_{2-\frac{1}{n}}$, and the other three irreducible representations up unitary equivalence are obtained via cyclic permutations of generators. The rank values are determined by comparing $[\pi]$ and $[\Phi^+(\pi)]$.
\end{proof}

Projections $\fP^{(n)}_i$ are central to the self-testing results in this paper. The intuition behind their applicability to self-tests is the following: if we momentarily forget irreducibility, they are characterized by having certain traces and satisfying a linear equation. In a quantum strategy with a maximally entangled state and projective measurements, traces and linear relations among the PVMs are encoded by the correlation. This makes strategies with maximally entangled states and measurements $(\fP^{(n)}_i,I-\fP^{(n)}_i)$ very natural candidates for the self-testing phenomenon.

\begin{rem}\label{r:recursive}
Proposition {\ref{p:ukr}} does not provide a closed-form expression for projections $\fP^{(n)}_1,\dots,\fP^{(n)}_4\in\mtxr{n}$ as functions of $n$. Nevertheless, definitions of functors $T$ and $S$ give rise to a recursive procedure for constructing $\fP^{(n)}_i\in\mtxr{n}$ from $\fP^{(n-1)}_i\in\mtxr{n-1}$.
This procedure requires only matrix arithmetic and Gram-Schmidt orthogonalization.
\\
Basis of recursion $n=1$: set $\fP^{(1)}_1:=1$ and $\fP^{(1)}_i:=0$ for $i=2,3,4$.
\\
Recursive step $n\to n+1$: given $\fP^{(n)}_1,\dots,\fP^{(n)}_4$ let
\begin{itemize}
	\item $U_i$ be an $n\times\rk (n-\fP^{(n)}_i)$ matrix whose columns form an orthonormal basis of the column space of $I-\fP^{(n)}_i$;
	\item $V_i$ be an $(\rk\fP^{(n)}_i) \times(n+1)$ matrix such that the columns of
	$$\begin{pmatrix}V_1 \\ \vdots \\ V_4
	\end{pmatrix}$$ 
form an orthonormal basis of the column space of
	$$I-\frac{1}{2+\frac{1}{n}}
	\begin{pmatrix}U_1^* \\ \vdots \\ U_4^*\end{pmatrix}
	\begin{pmatrix}U_1 &\cdots & U_4\end{pmatrix}.
	$$
\end{itemize}
Then set $\fP^{(n+1)}_i:=(2-\frac{1}{n+1})V_i^*V_i$.

Using the above procedure, we obtain the following projections for $n=1,2,3$:
\begin{align*}
&\fP^{(1)}_1=(1),\ \fP^{(1)}_2=(0),\ \fP^{(1)}_3=(0),\ \fP^{(1)}_4=(0) \\
&\fP^{(2)}_1=\begin{pmatrix}0&0\\0&0\end{pmatrix},\ 
\fP^{(2)}_2=\begin{pmatrix} 1 & 0 \\ 0 & 0\end{pmatrix},\ 
\fP^{(2)}_3=\begin{pmatrix}
	\frac{1}{4} & \frac{-\sqrt{3}}{4} \\
	\frac{-\sqrt{3}}{4} & \frac{3}{4} \\
\end{pmatrix},\ 
\fP^{(2)}_4=\begin{pmatrix}
	\frac{1}{4} & \frac{\sqrt{3}}{4} \\
	\frac{\sqrt{3}}{4} & \frac{3}{4} \\
\end{pmatrix} \\
&\fP^{(3)}_1=\begin{pmatrix}
	1 & 0 & 0 \\
	0 & 1 & 0 \\
	0 & 0 & 0
\end{pmatrix},\ 
\fP^{(3)}_2=\begin{pmatrix}
	0 & 0 & 0 \\
	0 & \frac{4}{9} & \frac{-2 \sqrt{5}}{9} \\
	0 & \frac{-2 \sqrt{5}}{9} & \frac{5}{9}
\end{pmatrix},\ \\
&\qquad\fP^{(3)}_3=\begin{pmatrix}
	\frac{1}{3} & \frac{1}{3 \sqrt{3}} & \frac{\sqrt{5}}{3\sqrt{3}} \\
	\frac{1}{3 \sqrt{3}} & \frac{1}{9} & \frac{\sqrt{5}}{9} \\
	\frac{\sqrt{5}}{3\sqrt{3}} & \frac{\sqrt{5}}{9} & \frac{5}{9}
\end{pmatrix},\ 
\fP^{(3)}_4=\begin{pmatrix}
	\frac{1}{3} & \frac{-1}{3 \sqrt{3}} & \frac{-\sqrt{5}}{3\sqrt{3}} \\
	\frac{-1}{3 \sqrt{3}} & \frac{1}{9} & \frac{\sqrt{5}}{9} \\
	\frac{-\sqrt{5}}{3\sqrt{3}} & \frac{\sqrt{5}}{9} & \frac{5}{9}
\end{pmatrix}
\end{align*}
The linear-algebraic nature of this procedure allows for a feasible implementation using exact arithmetic. For concrete matrices in cases $n=4,5,6$, see Appendix {\ref{app}}.
\end{rem}

For later use we record a technical fact.

\begin{lem}\label{l:invert}
The $4\times 4$ matrix
$$\begin{pmatrix}
\rk \fP^{(n)}_1 & \rk \fP^{(n)}_2 & \rk \fP^{(n)}_3 & \rk \fP^{(n)}_4 \\
\rk \fP^{(n)}_4 & \rk \fP^{(n)}_1 & \rk \fP^{(n)}_2 & \rk \fP^{(n)}_3 \\
\rk \fP^{(n)}_3 & \rk \fP^{(n)}_4 & \rk \fP^{(n)}_1 & \rk \fP^{(n)}_2 \\
\rk \fP^{(n)}_2 & \rk \fP^{(n)}_3 & \rk \fP^{(n)}_4 & \rk \fP^{(n)}_1
\end{pmatrix}=-(-1)^n I_4+\left\lfloor\frac{n}{2}\right\rfloor \begin{pmatrix}
1&1&1&1\\1&1&1&1\\1&1&1&1\\1&1&1&1
\end{pmatrix}
$$
is invertible for every $n\in\N$.
\end{lem}

\begin{rem}\label{r:trace}
Let us determine the normalized traces of $\fP^{(n)}_i$ and their products; these values will appear in the self-testing correlations of this paper. Clearly,
$$
\tau\left(\fP^{(n)}_1\right)=\frac12-\frac{1+3(-1)^n}{4n},\qquad
\tau\left(\fP^{(n)}_i\right)=\frac12-\frac{1-(-1)^n}{4n},\quad \text{for }i=2,3,4.
$$
Next, by Proposition \ref{p:ukr},
for every permutation $\sigma$ of $\{2,3,4\}$ there exists a unitary $U\in\mtxc{n}$ such that
$$U\fP^{(n)}_1U^*=\fP^{(n)}_1,\qquad 
U\fP^{(n)}_iU^*=\fP^{(n)}_{\sigma(i)},\quad
\text{for }i=2,3,4.
$$
Therefore $\tau(\fP^{(n)}_1\fP^{(n)}_i)$ is independent of $i\in\{2,3,4\}$, and $\tau(\fP^{(n)}_i\fP^{(n)}_j)$ is independent of $i,j\in\{2,3,4\}$ with $i\neq j$. From the equation $\sum_{j=1}^4\fP^{(n)}_i\fP^{(n)}_j=(2-\frac1n)\fP^{(n)}_i$ for $i=1,\dots,4$ we then obtain
\begin{align*}
\tau\left(\fP^{(n)}_1\fP^{(n)}_i\right)&=
\frac13\left(1-\frac1n\right)\tau\left(\fP^{(n)}_1\right)
\quad \text{for }i=2,3,4, \\
\tau\left(\fP^{(n)}_i\fP^{(n)}_j\right)&=
\frac12\left(1-\frac1n\right)\left(\tau\left(\fP^{(n)}_i\right)-\frac13\tau\left(\fP^{(n)}_1\right)\right)
\quad \text{for }i,j=2,3,4 \text{ and }i\neq j.
\end{align*}
\end{rem}

\section{Spectral results}\label{s:spec}

Let $n\in\N$. The projections $\fP^{(n)}_1,\dots,\fP^{(n)}_4$ of Proposition {\ref{p:ukr}} play a central role in self-tests of Section {\ref{s:st}} below. Namely, they appear as projective measurements in a self-tested strategy in Subsection {\ref{ss:state}}; the fact that they are determined by a linear relation $\fP^{(n)}_1+\cdots+\fP^{(n)}_4=(2-\frac1n)I$ is beneficial for deducing the measurements from the correlation. Nevertheless, to obtain a self-test, one still needs to be able to deduce the quantum state from the correlation. Furthermore, in Subsection {\ref{ss:meas}}, the presented strategies contain an additional projective measurement, which, while related to the $\fP^{(n)}_i$, is itself not a part of quadruple adding to a scalar multiple of identity. To help with the identification of the quantum state and the additional measurements from the correlation, we first require some information on eigenvalues and eigenvectors of certain tensor combinations and sums of pairs of the matrices $\fP^{(n)}_i$. Concretely, Proposition {\ref{p:eigvec}} shows how the maximally entangled state is related to $\fP^{(n)}_1,\dots,\fP^{(n)}_4$, and Proposition {\ref{p:spectrum}} shows that $\fP^{(n)}_3+\fP^{(n)}_4$ has pairwise distinct eigenvalues, which enables post-hoc self-testing techniques \mbox{\cite{supic,CMV}}.

\subsection{Role of the maximally entangled state}

First, we identify the largest eigenvalue of $\sum_i\fP^{(n)}_i\otimes \fP^{(n)}_i$ and the corresponding eigenvector (cf. \cite[Lemma 5.7]{mancinska}), and bound the spectrum of $\sum_i\fP^{(n)}_i\otimes \fP^{(n)}_{\sigma(i)}$ for a nontrivial cyclic permutation $\sigma$ of $(1,2,3,4)$.
Given $\ket{\psi}=\sum_{i,j}\alpha_{ij}\ket{i}\!\ket{j}\in\C^n\otimes\C^n$ let $\mat(\ket{\psi})=\sum_{i,j}\alpha_{ij}\ket{i}\!\bra{j}\in\mtxc{n}$ denote its matricization; note that $\mat(\ket{\phi_n})=\frac{1}{\sqrt{n}}I$, and
$$\mat\big(A\otimes B\ket{\psi}\big)=A\mat(\ket{\psi})B^\ti$$
for $A,B\in\mtxc{n}$.

\begin{lem}\label{l:eigvec}
Let $n\in\N$ and let $\sigma$ be a cyclic permutation $\sigma$ of $(1,2,3,4)$. Denote $M=\frac{n}{2n-1}\sum_{i=1}^4 \fP^{(n)}_i\otimes\fP^{(n)}_{\sigma(i)}$.
\begin{enumerate}[(i)]
    \item If $\sigma=\id$, then the largest eigenvalue of $M$ is 1, with the eigenspace $\C\ket{\phi_n}$.
    \item If $\sigma\neq\id$, then all eigenvalues of $M$ are strictly smaller than 1.
\end{enumerate}
\end{lem}

\begin{proof}
Let $\ket{\psi}\in\C^n\otimes\C^n$ be an arbitrary state. Then
\begin{equation}\label{e:eigest}
\begin{split}
\bra{\psi}I\otimes I-M\ket{\psi}
&\ge\bra{\psi}I\otimes I-\frac{n}{2n-1}\sum_{i=1}^4 \fP^{(n)}_i\otimes I\ket{\psi} \\
&=\bra{\psi}\left(I-\frac{n}{2n-1}\sum_{i=1}^4 \fP^{(n)}_i\right)\otimes I\ket{\psi}=0.
\end{split}
\end{equation}
Therefore the largest eigenvalue of $M$ is at most 1.
Since
$$
\bra{\phi_n}I\otimes I-\frac{n}{2n-1}\sum_{i=1}^4 \fP^{(n)}_i\otimes\fP^{(n)}_i\ket{\phi_n}=
\tau\left(I-\frac{n}{2n-1}\sum_{i=1}^4 \fP^{(n)}_i\right)=0,
$$
$\ket{\phi_n}$ is an eigenvector of $M$ for eigenvalue 1 if $\sigma=\id$.
Suppose $\ket{\psi}\in\C^n\otimes\C^n$ satisfies $M\ket{\psi}=\ket{\psi}$. Then \eqref{e:eigest} gives
$$\bra{\psi}M\ket{\psi}
=\bra{\psi}\frac{n}{2n-1}\sum_{i=1}^4 \fP^{(n)}_i\otimes I\ket{\psi}
$$
and therefore
$$\bra{\psi}\sum_{i=1}^4 \fP^{(n)}_i\otimes (I-\fP^{(n)}_{\sigma(i)})\ket{\psi}
=0.$$
Positive semidefinitness then implies $\fP^{(n)}_i\otimes (I-\fP^{(n)}_{\sigma(i)})\ket{\psi}=0$, and analogously
$(I-\fP^{(n)}_i)\otimes\fP^{(n)}_{\sigma(i)}\ket{\psi}=0$. In particular, $\fP^{(n)}_i\otimes I\ket{\psi}=I\otimes \fP^{(n)}_{\sigma(i)}\ket{\psi}$ for $i=1,\dots,4$. Therefore 
\begin{equation}\label{e:inter}
\fP^{(n)}_i\mat(\ket{\psi})=
\mat(\ket{\psi}) \fP^{(n)}_{\sigma(i)}
\qquad \text{for }i=1,\dots,4.
\end{equation}
Note that $\fP^{(n)}_1,\dots,\fP^{(n)}_4$ and $\fP^{(n)}_{\sigma(1)},\dots,\fP^{(n)}_{\sigma(4)}$ give rise to two irreducible representations of $\cA_{2-\frac1n}$ by Proposition \ref{p:ukr}, which are unitarily equivalent if and only if $\sigma=\id$. Since $\mat(\ket{\psi})$ intertwines these two irreducible representations, Schur's lemma implies that $\mat{\ket{\psi}}=\gamma I$ for some $\gamma\in\C$ if $\sigma=\id$, and $\mat{\ket{\psi}}=0$ if if $\sigma\neq\id$. 
Therefore $\ket{\psi}$ is a scalar multiple of $\ket{\phi_n}$ if $\sigma= \id$, and $1$ is not an eigenvalue of $M$ if $\sigma\neq \id$.
\end{proof}

The following proposition shows how the maximally entangled state $\ket{\phi_n}$ is intrinsically connected to representations of $\cA_{2-\frac1n}$.

\begin{prop}\label{p:eigvec}
Let $n\in\N$, let $a_1,\dots,a_4,b_1,\dots,b_4$ be nonnegative integers with $a_1+\cdots+a_4=b_1+\cdots+b_4$, and let $\sigma_1,\dots,\sigma_4$ be the distinct cyclic permutations of $(1,2,3,4)$. 
Consider the identification
$$\C^{(a_1+\cdots+a_4)n}\otimes \C^{(b_1+\cdots+b_4)n} \equiv 
\left(\bigoplus_{j,k=1}^4
\C^{a_j}\otimes \C^{b_k}\right)
\otimes (\C^n\otimes\C^n).$$
Then the largest eigenvalue of
$$\frac{n}{2n-1}\sum_{i=1}^4
\left(
\bigoplus_{j=1}^4 I_{a_j}\otimes\fP^{(n)}_{\sigma_j(i)}
\right)\otimes
\left(
\bigoplus_{j=1}^4 I_{b_j}\otimes
\fP^{(n)}_{\sigma_j(i)}
\right)$$
is 1, with the eigenspace 
$$\Big\{\big(
\ket{\aux_1}\oplus\ket{\aux_2}\oplus\ket{\aux_3}\oplus\ket{\aux_4}
\big)\otimes\ket{\phi_n}
\colon \ket{\aux_j}\in\C^{a_j}\otimes \C^{b_j}\Big\}.$$
\end{prop}

\begin{proof}
Follows from the distributivity of tensor product over direct sum, and Lemma \ref{l:eigvec}.
\end{proof}

\subsection{Spectrum of the sum of two distinguished projections}

Next, we analyze the spectrum of the matrix $\fP^{(n)}_3+\fP^{(n)}_4$ for every $n$. To do this, we return to the functors between categories $\Rep_\alpha$.
Given a finite-dimensional representation $\pi$ of $\cA_\alpha$, let  $\Lambda_\pi\subset[0,2]$ denote the set of eigenvalues of $\pi(x_3+x_4)$. 

\begin{lem}\label{l:spec_funct}
Let $\pi$ be an $n$-dimensional representation of $\cA_\alpha$.
\begin{enumerate}[(i)]
    \item $\Lambda_{T(\pi)}=2-\Lambda_\pi$.
    \item Let $\alpha\notin\{0,1\}$. 
    \begin{enumerate}
    \item[(ii.a)] If $\rk\pi(x_1)+\rk\pi(x_2)>n= \rk\pi(x_3)+\rk\pi(x_4)$ then $$\Lambda_{S(\pi)}=\{0\}\cup\left(\tfrac{\alpha}{\alpha-1}-\tfrac{1}{\alpha-1}\Lambda_\pi\right).$$
    \item[(ii.b)] If $\rk\pi(x_3)+\rk\pi(x_4)>n= \rk\pi(x_1)+\rk\pi(x_2)$ then $$\Lambda_{S(\pi)}=\left\{\tfrac{\alpha}{\alpha-1}\right\}\cup\left(\tfrac{\alpha}{\alpha-1}-\tfrac{1}{\alpha-1} \Lambda_\pi\right).$$
    \end{enumerate}
    \item Let $\alpha\in(0,3)$. 
    \begin{enumerate}
    \item[(iii.a)] If $\rk\pi(x_1)+\rk\pi(x_2)<n= \rk\pi(x_3)+\rk\pi(x_4)$ then $$\Lambda_{\Phi^+(\pi)}=\{0\}\cup\left(1-\tfrac{1}{3-\alpha}+\tfrac{1}{3-\alpha}\Lambda_\pi\right).$$
    \item[(iii.b)] If $\rk\pi(x_3)+\rk\pi(x_4)<n= \rk\pi(x_1)+\rk\pi(x_2)$ then $$\Lambda_{\Phi^+(\pi)}=\left\{1+\tfrac{1}{3-\alpha}\right\}\cup\left(1-\tfrac{1}{3-\alpha}+\tfrac{1}{3-\alpha}\Lambda_\pi\right).$$
    \end{enumerate}
\end{enumerate}
\end{lem}

\begin{proof}
Equation (i) follows immediately from $T(\pi)(x_i)=I-\pi(x_i)$. Equations (iii) are consequences of (i) and (ii) because $\Phi^+=S\circ T$.

Equations (ii): Suppose $\pi$ act on $\cH$ with $\dim\cH=n$, and let
\begin{align*}
u_i&:\ran\pi(x_i)\to\cH, \\
w_i&:\ran\pi(x_i)\to \ran\pi(x_1)\oplus\cdots\oplus \ran\pi(x_4),\\
v=\left(\begin{smallmatrix}
v_1\\ \vdots\\v_4
\end{smallmatrix}\right)&: \ran\left(
I-\frac{1}{\alpha}\left(\begin{smallmatrix}
u_1^* \\ \vdots \\ u_4^*
\end{smallmatrix}\right)\left(\begin{smallmatrix}
u_1 & \cdots & u_4
\end{smallmatrix}\right)
\right)\to \ran\pi(x_1)\oplus\cdots\oplus \ran\pi(x_4)
\end{align*}
be inclusions as in the construction of $S$. 
Then $S(\pi)(x_i)=\frac{\alpha}{\alpha-1}v^*w_iw_i^*v$, and the characteristic polynomial of $S(\pi)(x_3+x_4)$ equals
\begin{align*}
&\det\big(\lambda I-S(\pi)(x_3+x_4)\big)\\
=\ &\det\left(\lambda I-\tfrac{\alpha}{\alpha-1}v^*(w_3w_3^*+w_4w_4^*)v\right) \\
=\ &\det\left(\lambda I-\tfrac{\alpha}{\alpha-1}\left(\begin{smallmatrix}
v_3^*&v_4^*
\end{smallmatrix}\right)\left(\begin{smallmatrix}
v_3\\ v_4
\end{smallmatrix}\right)\right) \\
=\ &\lambda^{\rk\pi(x_1)+\rk\pi(x_2)-n}\det\left(\lambda I-\tfrac{\alpha}{\alpha-1}\left(\begin{smallmatrix}
v_3\\ v_4
\end{smallmatrix}\right)\left(\begin{smallmatrix}
v_3^*&v_4^*
\end{smallmatrix}\right)\right) \\
=\ &\lambda^{\rk\pi(x_1)+\rk\pi(x_2)-n}\det\left(\lambda I-\tfrac{\alpha}{\alpha-1}\left(I-\tfrac{1}{\alpha}\left(\begin{smallmatrix}
u_3^*\\ u_4^*
\end{smallmatrix}\right)\left(\begin{smallmatrix}
u_3&u_4
\end{smallmatrix}\right)\right)\right) \\
=\ &\lambda^{\rk\pi(x_1)+\rk\pi(x_2)-n}\det\left(\big(\lambda-\tfrac{\alpha}{\alpha-1}\big) I+\tfrac{1}{\alpha-1}\left(\begin{smallmatrix}
u_3^*\\ u_4^*
\end{smallmatrix}\right)\left(\begin{smallmatrix}
u_3&u_4
\end{smallmatrix}\right)\right) \\
=\ &\lambda^{\rk\pi(x_1)+\rk\pi(x_2)-n}\left(\lambda-\tfrac{\alpha}{\alpha-1}
\right)^{\rk\pi(x_3)+\rk\pi(x_4)-n}\det\left(\big(\lambda-\tfrac{\alpha}{\alpha-1}\big) I+\tfrac{1}{\alpha-1}\left(\begin{smallmatrix}
u_3&u_4
\end{smallmatrix}\right)\left(\begin{smallmatrix}
u_3^*\\ u_4^*
\end{smallmatrix}\right)\right) \\
=\ &\lambda^{\rk\pi(x_1)+\rk\pi(x_2)-n}\left(\lambda-\tfrac{\alpha}{\alpha-1}
\right)^{\rk\pi(x_3)+\rk\pi(x_4)-n}\det\left(\big(\lambda-\tfrac{\alpha}{\alpha-1}\big) I+\tfrac{1}{\alpha-1}\pi(x_3+x_4)\right).
\end{align*}
Therefore
$$\Lambda_{S(\pi)}=\{0\}\cup\big(\tfrac{\alpha}{\alpha-1}-\tfrac{1}{\alpha-1}\Lambda_\pi\big)$$
if $\rk\pi(x_1)+\rk\pi(x_2)>n=\rk\pi(x_3)+\rk\pi(x_4)$, and
$$\Lambda_{S(\pi)}=\left\{\tfrac{\alpha}{\alpha-1}\right\}\cup\big(\tfrac{\alpha}{\alpha-1}-\tfrac{1}{\alpha-1}\Lambda_\pi\big)$$
if $\rk\pi(x_3)+\rk\pi(x_4)>n=\rk\pi(x_1)+\rk\pi(x_2)$.
\end{proof}

The following proposition identifies all eigenvalues of the matrix $\fP_3^{(n)}+\fP_4^{(n)}$; in particular, they are all simple (pairwise distinct).

\begin{prop}\label{p:spectrum}
Eigenvalues of $n(\fP_3^{(n)}+\fP_4^{(n)})$
are $\{0,2,\dots,2n-2\}$ if $n$ is odd, and $\{1,3,\dots,2n-1\}$ if $n$ is even.
\end{prop}

\begin{proof}
Let $\pi_1:\cA_1\to\C$ be given as $\pi_1(x_1)=1$ and $\pi_1(x_2)=\pi(x_3)=\pi(x_4)=0$. For $n\ge 2$ denote $\pi_n=\Phi^+(\pi_1)$. By Proposition \ref{p:ukr} we have $\rk \pi_n(x_1)+\rk \pi_n(x_2)<n=\rk \pi_n(x_3)+\rk \pi_n(x_4)$ if $n$ is even, and $\rk \pi_n(x_3)+\rk \pi_n(x_4)<n=\rk \pi_n(x_1)+\rk \pi_n(x_2)$ if $n$ is odd.
By Lemma \ref{l:spec_funct},
\begin{alignat*}{2}
\Lambda_{\pi_{n+1}}&=\{0\}\cup\left(
\tfrac{1}{n+1}+\tfrac{n}{n+1}\Lambda_{\pi_n}
\right) \qquad &&\text{if }n \text{ is even,}\\
\Lambda_{\pi_{n+1}}&=\{2-\tfrac{1}{n+1}\}\cup\left(
\tfrac{1}{n+1}+\tfrac{n}{n+1}\Lambda_{\pi_n}
\right) \qquad &&\text{if }n \text{ is odd.}
\end{alignat*}
Therefore
\begin{alignat*}{2}
(n+1)\Lambda_{\pi_{n+1}}&=\{0\}\cup\left(
1+n\Lambda_{\pi_n}
\right) \qquad &&\text{if }n \text{ is even,}\\
(n+1)\Lambda_{\pi_{n+1}}&=\{2n+1\}\cup\left(
1+n\Lambda_{\pi_n}
\right) \qquad &&\text{if }n \text{ is odd.}
\end{alignat*}
Since $\Lambda_{\pi_1}=\{0\}$, induction on $n$ shows that
\begin{alignat*}{2}
n\Lambda_{\pi_n}&=\{0,2,\dots,2n-2\}\qquad &&\text{if }
n \text{ is odd,}\\
n\Lambda_{\pi_n}&=\{1,3,\dots,2n-1\}\qquad &&\text{if }
n \text{ is even.}
\end{alignat*}
Finally, $\fP_3^{(n)},\fP_4^{(n)}$ are simultaneously unitarily equivalent to $\pi_n(x_3),\pi_n(x_4)$.
\end{proof}

Lastly, we determine how eigenvectors of $\fP_3^{(n)}+\fP_4^{(n)}$ interact with $\fP_1^{(n)}$ and $\fP_2^{(n)}$.

\begin{prop}\label{p:34to12}
Let $\lambda$ be an  eigenvalue of $\fP_3^{(n)}+\fP_4^{(n)}$, with a corresponding unit eigenvector $\ket{e}\in\R^n$.
\begin{enumerate}[(i)]
\item If $\lambda\neq1-\frac1n$ then
$$\bra{e}\fP^{(n)}_1\ket{e}=\bra{e}\fP^{(n)}_2\ket{e}
=1-\frac{1}{2n}-\frac{\lambda}{2}.$$
\item If $\lambda=1-\frac1n$ then
$$\bra{e}\fP^{(n)}_1\ket{e}=\left\{
\begin{array}{ll}
0 &  \text{if }n\text{ even,}\\
1 &  \text{if }n\text{ odd,}
\end{array}
\right.\qquad
\bra{e}\fP^{(n)}_2\ket{e}=\left\{
\begin{array}{ll}
1 &  \text{if }n\text{ even,}\\
0 &  \text{if }n\text{ odd.}
\end{array}
\right.
$$
\end{enumerate}
\end{prop}

\begin{proof}
(i) By the defining relation of $\fP^{(n)}_i$,
\begin{equation}\label{e:releig}
\fP^{(n)}_1\ket{e}+\fP^{(n)}_2\ket{e}+\lambda\ket{e}=
\left(2-\frac1n\right)\ket{e}.
\end{equation}
Multiplying \eqref{e:releig} on the left with $\bra{e}\fP^{(n)}_i$ for $i=1,2$ results in
\begin{align*}
\bra{e}\fP^{(n)}_1\ket{e}
+\bra{e}\fP^{(n)}_1\fP^{(n)}_2\ket{e}&=
\left(2-\frac1n-\lambda\right)\bra{e}\fP^{(n)}_1\ket{e},
\\
\bra{e}\fP^{(n)}_2\fP^{(n)}_1\ket{e}
+\ket{e}\fP^{(n)}_2\ket{e}&=
\left(2-\frac1n-\lambda\right)\bra{e}\fP^{(n)}_2\ket{e}.
\end{align*}
Therefore $\bra{e}\fP^{(n)}_1\ket{e}=\bra{e}\fP^{(n)}_2\ket{e}$ if $\lambda\neq1-\frac1n$.
Multiplying \eqref{e:releig} on the left with $\bra{e}$ then gives $\bra{e}\fP^{(n)}_1\ket{e}=\bra{e}\fP^{(n)}_2\ket{e}=1-\frac{1}{2n}-\frac{\lambda}{2}$.

(ii) Note that $\fP_3^{(n)}+\fP_4^{(n)}$ admits $n$ orthonormal eigenvectors $\ket{e_1},\dots,\ket{e_n}\in\R^n$ by Proposition \ref{p:spectrum}. Hence
$$\tr \fP_i^{(n)} = \sum_{k=1}^n\bra{e_k}\fP^{(n)}_i\ket{e_k}$$
for $i=1,2$. By (ii) and Proposition \ref{p:ukr} we therefore have
\begin{align*}
\bra{e}\fP^{(n)}_i\ket{e}&
=\tr \fP_i^{(n)}-(n-1)\left(1-\frac{1}{2n}\right)+\frac12\left(
\tr \Big(\fP_3^{(n)}+ \fP_4^{(n)}\Big)-1+\frac1n
\right)\\
&=2\left\lfloor\frac{n}{2}\right\rfloor-n+1-\left\{
\begin{array}{ll}
(-1)^n &  \text{if }i=1\\
0 &   \text{if }i=2
\end{array}
\right.
\end{align*}
since $\tr \fP^{(n)}_{i}=\rk \fP^{(n)}_{i}$.
\end{proof}

\section{Constant-sized self-tests}\label{s:st}

In this section we derive the main results of the paper: every maximally entangled state is self-tested by a 4-input 2-output strategy (Subsection \ref{ss:state}), and every single binary PVM is self-tested by a 5-input 2-output strategy (Subsection \ref{ss:meas}).

\subsection{Self-testing maximally entangled states}\label{ss:state}

First we introduce a family of 4-input 2-output strategies that self-test maximally entangled states of all dimensions (Theorem \ref{t:Sn_self_test}).

\begin{defn}\label{d:4in}
For $n\in\N$ let $\fP^{(n)}_i$ be the $n\times n$ projections as in Proposition \ref{p:ukr}. Let $\cS_n$ be the 4-input 2-output bipartite strategy
$$\cS_n=\left(
\ket{\phi_n};
\left(\fP^{(n)}_i,I-\fP^{(n)}_i\right)_{i=1}^4;
\left(\fP^{(n)}_i,I-\fP^{(n)}_i\right)_{i=1}^4
\right).$$
\end{defn}

Note that the correlation of $\cS_{n}$ is synchronous,  $p(a,b|i,i)=\tau\left(\fP^{(n)}_i(I-\fP^{(n)}_i)\right)=0$ for $a\neq b$.
Furthermore,
\begin{align*}
	p(1,1|i,j)&=\bra{\phi_n}\fP^{(n)}_i\otimes \fP^{(n)}_j\ket{\phi_n} =\tau\left(\fP^{(n)}_i\fP^{(n)}_j\right),\\
	p(1|i)&=\bra{\phi_n}\fP^{(n)}_i\otimes I\ket{\phi_n}=\bra{\phi_n}I\otimes \fP^{(n)}_i\ket{\phi_n}=\tau\left(\fP^{(n)}_i\right)
\end{align*}
for $i,j=1,\dots,4$, and these values are computed in Remark \ref{r:trace}. 
Comprising everything together, the correlation of $\cS_{n}$ is determined by the vector
$$\big(p(1|i)\big)_{i=1}^4 = \begin{pmatrix}
	\frac{\lfloor\frac{n}{2}\rfloor-(-1)^n}{n} &
	\frac{\lfloor\frac{n}{2}\rfloor}{n} &
	\frac{\lfloor\frac{n}{2}\rfloor}{n} &
	\frac{\lfloor\frac{n}{2}\rfloor}{n}
\end{pmatrix}$$
and the symmetric matrix $\big(p(1,1|i,j)\big)_{i,j=1}^4$
$$
\left(\begin{array}{cccc}
	\frac{\lfloor\frac{n}{2}\rfloor-(-1)^n}{n} & \frac{(n-1)(\lfloor\frac{n}{2}\rfloor-(-1)^n)}{3n^2} & \frac{(n-1)(\lfloor\frac{n}{2}\rfloor-(-1)^n)}{3n^2} & \frac{(n-1)(\lfloor\frac{n}{2}\rfloor-(-1)^n)}{3n^2} \\[5pt]
	\cdot &
	\frac{\lfloor\frac{n}{2}\rfloor}{n} &
	\frac{(n-1)(2n-1+3(-1)^n)}{12n^2} &
	\frac{(n-1)(2n-1+3(-1)^n)}{12n^2} \\[5pt]
	\cdot &
	\cdot &
	\frac{\lfloor\frac{n}{2}\rfloor}{n} &
	\frac{(n-1)(2n-1+3(-1)^n)}{12n^2} \\[5pt]
	\cdot & \cdot & \cdot &
	\frac{\lfloor\frac{n}{2}\rfloor}{n} \\
\end{array}\right).
$$
Notice that while a closed-form expression for the strategy $\cS_n$ has not been given (instead, the projections in $\cS_n$ can be recursively constructed as in Remark {\ref{r:recursive}}), its correlation admits a closed-form expression (as a function of $n$).

The next theorem establishes that $\cS_n$ is a local dilation of any strategy $\cS$ that produces the same correlation as $\cS_n$. The blueprint for the proof is threefold. Firstly, the correlation manages to encode the defining linear relation of measurements in $\cS_n$, which leads to measurements of $\cS$ essentially forming a representation of $\cA_{2-\frac1n}$. Secondly, the established relationship between the maximally entangled state and representations of $\cA_{2-\frac1n}$ (Proposition {\ref{p:eigvec}}) allows one to identify the state in $\cS$. Thirdly, the finer look at the correlation shows that the representation of $\cA_{2-\frac1n}$  arising from measurements of $\cS$ cannot be an direct sum of the different irreducible representations, but is actually a direct copy of the irreducible representation coming from $\cS_n$.

\begin{thm}\label{t:Sn_self_test}
The strategy $\cS_n$ is self-tested by its correlation for every $n\in\N$.
\end{thm}

\begin{proof}
Let $p$ be the correlation of $\cS_n$.
Suppose
$$\cS=\left(
\ket{\psi};(P_i,I-P_i)_{i=1}^4;(Q_i,I-Q_i)_{i=1}^4
\right)$$
is another strategy with the correlation $p$. Since $p$ is synchronous and local dilations are transitive, by \cite[Lemma 4.9 and Corollary 3.6]{mancinska} it suffices to assume that the state $\ket{\psi}\in\cH\otimes \cH$ has full Schmidt rank, $P_i,Q_i$ are projections on $\cH$, and
\begin{equation}\label{e:LR}
P_i\otimes I\ket{\psi} = I\otimes Q_i\ket{\psi}
\end{equation}
for $i=1,\dots,4$.
By equality of correlations and \eqref{e:LR},
\begin{align*}
&\bra{\psi}\left(\frac{2n-1}{n}I-\sum_{i=1}^4P_i\right)^2\otimes I\ket{\psi}\\
=\,&\bra{\psi}\left(\frac{2n-1}{n}I-\sum_{i=1}^4P_i\right)\otimes \left(\frac{2n-1}{n}I-\sum_{i=1}^4Q_i\right)\ket{\psi} \\
=\,&\bra{\phi_n}\left(\frac{2n-1}{n}I-\sum_{i=1}^4\fP^{(n)}_i\right)\otimes \left(\frac{2n-1}{n}I-\sum_{i=1}^4\fP^{(n)}_i\right)\ket{\phi_n} =0,
\end{align*}
and analogously for $Q_i$.
Since $\ket{\psi}$ has full rank, we obtain
\begin{equation}\label{e:idalg}
\frac{2n-1}{n}I-\sum_{i=1}^4P_i=0=
\frac{2n-1}{n}I-\sum_{i=1}^4Q_i.
\end{equation}
Furthermore,
\begin{equation}\label{e:ideig}
\bra{\psi}\frac{n}{2n-1}\sum_{i=1}^4P_i\otimes Q_i\ket{\psi} = 
\bra{\phi_n}\frac{n}{2n-1}\sum_{i=1}^4\fP^{(n)}_i\otimes \fP^{(n)}_i\ket{\phi_n} = 1.
\end{equation}
Let $\sigma_1,\dots,\sigma_4$ be the distinct cyclic permutations of $(1,2,3,4)$, with $\sigma_1=\id$.
By \eqref{e:idalg} and Proposition \ref{p:ukr} there exist nonnegative integers $a_1,\dots,a_4,b_1,\dots,b_4$ with $a_1+\cdots+a_4=b_1+\cdots+b_4$, and 
unitaries $U$ and $V$ on $\cH$, such that
$$
UP_iU^*=\bigoplus_{j=1}^4 I_{a_j}\otimes
\fP^{(n)}_{\sigma_j(i)},\qquad
VQ_iV^*=\bigoplus_{j=1}^4 I_{b_j}\otimes\fP^{(n)}_{\sigma_j(i)}
$$
for $i=1,\dots,4$.
By \eqref{e:ideig} and Proposition \ref{p:eigvec},
$$U\otimes V\ket{\psi}=\big(
\ket{\aux_1}\oplus\ket{\aux_2}\oplus\ket{\aux_3}\oplus\ket{\aux_4}
\big)\otimes \ket{\phi_n}$$
for some $\ket{\aux_j}\in\C^{a_j}\otimes \C^{b_j}$, where we identified
$$\cH\otimes\cH\equiv \left(\bigoplus_{j,k=1}^4
\C^{a_j}\otimes \C^{b_k}\right)
\otimes (\C^n\otimes\C^n).$$
Then
$$\bra{\phi_n}\fP^{(n)}_i\otimes I\ket{\phi_n} =\bra{\psi}P_i\otimes I\ket{\psi}
=\sum_{j=1}^4\braket{\aux_j}{\aux_j}\bra{\phi_n}\fP^{(n)}_{\sigma_j(i)}\otimes I\ket{\phi_n}$$
gives rise to a linear system of equations in $\braket{\aux_j}{\aux_j}$,
\begin{equation}\label{e:linsys}
\rk\fP^{(n)}_i
=\sum_{j=1}^4\rk\fP^{(n)}_{\sigma_j(i)}\cdot \braket{\aux_j}{\aux_j}
\qquad \text{for }i=1,2,3,4.
\end{equation}
By Lemma \ref{l:invert}, the system \eqref{e:linsys} has a unique solution; since $\sigma_1=\id$, we obtain $\braket{\aux_1}{\aux_1}=1$ and $\braket{\aux_j}{\aux_j}=0$ for $j=2,3,4$.
Since $\ket{\psi}$ is a faithful state, it follows that $a_j=b_j=0$ for $j=2,3,4$, and $a_1=b_1$.
Therefore
$$UP_iU^*=I_{a_1}\otimes\fP^{(n)}_i,\qquad
VQ_iV^*=I_{a_1}\otimes\fP^{(n)}_i,\qquad
U\otimes V\ket{\psi}=\ket{\aux_1}\otimes\ket{\phi_n},
$$
so $\cS_n$ is a local dilation of $\cS$.
\end{proof}

\begin{rem}\label{r:robust}
The proof of Theorem \ref{t:Sn_self_test} follows the core ideas of the proof of \cite[Corollary 7.1]{mancinska}, which treats maximally entangled states of odd dimension. The main difference arises from applying the representation theory of C*-algebras $\cA_\alpha$ for different values of $\alpha$. Namely, in \cite{mancinska} the authors focus on $\cA_{2-\frac2n}$ for odd $n$ (and their analogs on more than four generators), since $\cA_{2-\frac2n}$ for odd $n$ is simple and isomorphic to $\mtxc{n}$ (i.e., it has a unique irreducible representation, which is $n$-dimensional). On the other hand, algebras $\cA_{2-\frac1n}$ for $n\in\N$ are not simple, as they are isomorphic to $\C^4\otimes \mtxc{n}$.
Non-simplicity is the origin of intricacies in the proof of Theorem \ref{t:Sn_self_test} and auxiliary results.

Finally, with a considerable effort, the authors of \cite{mancinska} also establish that their self-tests are \emph{robust}. Such robustness analysis is omitted in this paper; nevertheless, there is no obstruction for the techniques of \cite[Section 6]{mancinska} to imply robust versions of the newly presented self-tests.
\end{rem}

\begin{cor}\label{c:st4}
The following states and binary projective measurements can be self-tested by 4-input 2-output bipartite strategies for every $n\in\N$:
\begin{enumerate}[(a)]
	\item maximally entangled state of local dimension $n$;
	\item binary projective measurement determined by an $n\times n$ projection with rank in
	$$\left\{
	\left\lceil\frac{n}{2}\right\rceil,\
	\left\lfloor\frac{n}{2}\right\rfloor-(-1)^n, \
	\left\lceil\frac{n}{2}\right\rceil+(-1)^n
	\right\}.$$
\end{enumerate}
\end{cor}

\subsection{Self-testing local projective measurements}\label{ss:meas}

Next we introduce a two-parametric family of 5-input 2-output strategies that self-test binary PVMs of all dimensions and ranks (Theorem {\ref{t:Snr_self_test}}).
These strategies are obtained from the 4-input 2-output strategies of Subsection {\ref{ss:state}} by adding an additional binary PVM. The phenomenon, where a self-tested strategy is extended to a new one while preserving the self-testing feature, is called post-hoc self-testing \mbox{\cite{supic}}. 
The key sufficiency condition for post-hoc self-testing was derived in \mbox{\cite{CMV}}, and is presented next.

Given an invertible hermitian matrix $X\in\mtxc{n}$ let $\sgn(X)\in\mtxc{n}$ be the unique hermitian unitary matrix that commutes with $X$, and $\sgn(X)X\succ0$. Equivalently, $\sgn(X)$ is the unitary part of the polar decomposition of $X$. In other words, $\sgn$ is the matrix extension of the usual sign function via functional calculus. This map plays a role in the following post-hoc self-testing criterion established in \cite{CMV}.

\begin{prop}{{\cite[Proposition 3.7]{CMV}}}\label{p:chen}
Suppose $P,P_i,Q_j\in\mtxr{n}$ for $i=1,\dots,N_A$ and $j=1,\dots,N_B$ are projections, and the $(N_A,N_B)$-input $(2,2)$-output strategy
$$\left(\ket{\phi_n};
\left(P_i,I-P_i\right)_{i=1}^{N_A}; 
\left(Q_i,I-Q_i\right)_{i=1}^{N_B}
\right)$$
is self-tested by its correlation. If
$$2P-I\in\sgn\Big(
\glr{n}\cap\spa_\R\{I,Q_1,\dots,Q_{N_B}\}
\Big),$$
then the $(N_A+1,N_B)$-input $(2,2)$-output  strategy
$$\left(\ket{\phi_n};
\left(P_i,I-P_i\right)_{i=1}^{N_A},(P,I-P); 
\left(Q_i,I-Q_i\right)_{i=1}^{N_B}
\right)$$
is self-tested by its correlation.
\end{prop}

As mentioned at the beginning of the subsection, Proposition {\ref{p:chen}} will be used to obtain a self-tested strategy by extending $\cS_n$ from Subsection {\ref{ss:state}}. Recall that $\fP^{(n)}_3+\fP^{(n)}_4$ has pairwise distinct eigenvalues by Proposition {\ref{p:spectrum}}. This gives rise to a family of projections that satisfy the sufficiency condition in Proposition {\ref{p:chen}}.

\begin{prop}\label{p:distinct}
Let $n,r\in\N$ with $r\le n$.
The matrix
$$\fQ^{(n,r)}:=\frac12\bigg(I+\sgn\Big((2r-\tfrac12)I-
n\big(\fP^{(n)}_3+\fP^{(n)}_4\big)
\Big)\bigg)\in\mtxr{n}$$
is a projection of rank $r$, and satisfies
$$2\fQ^{(n,r)}-I\in\sgn\Big(
\glr{n}\cap\spa_\R\{I,\fP^{(n)}_3,\fP^{(n)}_4\}
\Big).$$
\end{prop}

\begin{proof}
The matrix $\fQ^{(n,r)}$ is a projection by definition of the map $\sgn$.
By Proposition \ref{p:spectrum}, the matrix $n(\fP^{(n)}_3+\fP^{(n)}_4)$ has eigenvalues $\{0,2,\dots,2n-2\}$ if $n$ is odd and $\{1,3,\dots,2n-1\}$ if $n$ is even. 
Therefore $(2r-\tfrac12)I-
n(\fP^{(n)}_3+\fP^{(n)}_4)$ has $r$ positive eigenvalues and $n-r$ negative eigenvalues. 
Consequently, the multiplicities of eigenvalues $1$ and $-1$ of $\sgn((2r-\tfrac12)I-
n(\fP^{(n)}_3+\fP^{(n)}_4))$ are $r$ and $n-r$, respectively.
Hence the rank of $\fQ^{(n,r)}$ is $r$.
\end{proof}

\begin{rem}\label{r:Qalt}
For $r\le n$ let $\ket{e_1},\dots,\ket{e_r}\in\R^n$ be unit eigenvectors of $\fP^{(n)}_3+\fP^{(n)}_4$ corresponding to the smallest $r$ eigenvalues in increasing order (note that $\ket{e_i}$ are uniquely determined up to a sign because $\fP^{(n)}_3+\fP^{(n)}_4$ has $n$ distinct eigenvalues). Then $$\fQ^{(n,r)}=\ket{e_1}\!\bra{e_1}+\cdots+ \ket{e_r}\!\bra{e_r}.$$
For concrete matrix representations of $\fQ^{(n,r)}$ when $1\le r<n\le 6$, see Appendix \ref{app}. While this is arguably a simpler and computationally more available definition of $\fQ^{(n,r)}$ than the original in Proposition \ref{p:distinct}, the presentation in terms of the $\sgn$ map is critical in establishing the self-test of Theorem \ref{t:Snr_self_test} below.
\end{rem}

\begin{rem}\label{r:traceQ}
Let us determine the normalized traces of $\fP^{(n)}_i\fQ^{(n,r)}$ for $r<\frac{n}{2}$. Clearly, $\tau\left(\fQ^{(n,r)}\right)=\frac{r}{n}$.
By Proposition \ref{p:ukr} there exists a unitary $U\in\mtxc{n}$ such that $U\fP^{(n)}_3U^*=\fP^{(n)}_4$ and $U\fP^{(n)}_4U^*=\fP^{(n)}_3$, and therefore $\tr\left(\fP^{(n)}_3\fQ^{(n,r)}\right)
=\tr\left(\fP^{(n)}_4\fQ^{(n,r)}\right)$.
Thus
$$
\tau\left(\fP^{(n)}_i\fQ^{(n,r)}\right)=
\frac12\tau\left(\fQ^{(n,r)}\big(\fP^{(n)}_3+\fP^{(n)}_4\big)\fQ^{(n,r)}\right)=
\frac{r}{2n^2}\left(r-\frac{1-(-1)^n}{2}\right)
$$
for $i=3,4$ by Proposition \ref{p:spectrum}, since $\tr(\fQ^{(n,r)}(\fP^{(n)}_3+\fP^{(n)}_4)\fQ^{(n,r)})$ is the sum of smallest $r$ eigenvalues of $\fP^{(n)}_3+\fP^{(n)}_4$ by Remark \ref{r:Qalt}.
Since $r<\frac{n}{2}$, Proposition \ref{p:34to12} and Remark \ref{r:Qalt} imply $\tr(\fQ^{(n,r)}\fP^{(n)}_1\fQ^{(n,r)})=\tr(\fQ^{(n,r)}\fP^{(n)}_2\fQ^{(n,r)})$. By the defining relation of $\fP^{(n)}_i$ we then obtain
$$\tau\left(\fP^{(n)}_i\fQ^{(n,r)}\right)=
\frac12\left(
\left(2-\frac1n\right)\tau\left(\fQ^{(n,r)}\right)
-\tau\left(\fP^{(n)}_3\fQ^{(n,r)}\right)
-\tau\left(\fP^{(n)}_4\fQ^{(n,r)}\right)
\right)
$$
for $i=1,2$.
\end{rem}

\begin{defn}\label{d:5in}
Given $n,r\in\N$ with $r< n$, let $\fP^{(n)}_i$ be as in Proposition \ref{p:ukr}, and let $\fQ^{(n,r)}$ be as in Proposition \ref{p:distinct}. Let $\cS_{n,r}$ be the $(5,4)$-input $(2,2)$-output bipartite strategy
$$\left(\ket{\phi_n};
\left(\fP^{(n)}_i,I-\fP^{(n)}_i\right)_{i=1}^4,(\fQ^{(n,r)},I-\fQ^{(n,r)});
\left(\fP^{(n)}_i,I-\fP^{(n)}_i\right)_{i=1}^4
\right).
$$
Since $\cS_{n,r}$ is an extension of $\cS_n$, its correlation is determined by that of $\cS_n$ and
\begin{align*}
p(1|5)&=\bra{\phi_n}\fQ^{(n,r)}\otimes I\ket{\phi_n}=\tau\left(\fQ^{(n,r)}\right),\\
p(1,1|i,5)&=\bra{\phi_n}\fQ^{(n,r)}\otimes \fP^{(n)}_j\ket{\phi_n} =\tau\left(\fP^{(n)}_i\fQ^{(n,r)}\right)
\end{align*}
for $i=1,\dots,4$, which are computed in Remark \ref{r:traceQ}. 
\end{defn}

Let $n,r\in\N$ with $r<n$. If $r=\frac{n}{2}$, then a binary projective measurement of dimension $n$ and rank $r$ is up to a unitary basis change contained in the self-tested strategy $\cS_n$. Otherwise, a binary projective measurement of dimension $n$ and rank $r$ is contained, up to a unitary basis change and a reordering of outputs, in $\cS_{n,r}$ or $\cS_{n,n-r}$. For this reason, let us explicitly determine the correlation of $\cS_{n,r}$ only for $r<\frac{n}{2}$.
Since $\cS_{n,r}$ is an extension of $\cS_n$ (whose correlation is given in Subsection {\ref{ss:state}}) and Remark {\ref{r:traceQ}} computes the additional inner products (for $r<\frac{n}{2}$),
the correlation of $\cS_{n,r}$ is determined by the vector
$$\big(p(1|j)\big)_{j=1}^5 = \begin{pmatrix}
	\frac{\lfloor\frac{n}{2}\rfloor-(-1)^n}{n} &
	\frac{\lfloor\frac{n}{2}\rfloor}{n} &
	\frac{\lfloor\frac{n}{2}\rfloor}{n} &
	\frac{\lfloor\frac{n}{2}\rfloor}{n} & \frac{r}{n}
\end{pmatrix}$$
and the $5\times 4$ matrix $\big(p(1,1|i,j)\big)_{i,j}$
$$\left(\begin{array}{cccc}
	\frac{\lfloor\frac{n}{2}\rfloor-(-1)^n}{n} & \frac{(n-1)(\lfloor\frac{n}{2}\rfloor-(-1)^n)}{3n^2} & \frac{(n-1)(\lfloor\frac{n}{2}\rfloor-(-1)^n)}{3n^2} & \frac{(n-1)(\lfloor\frac{n}{2}\rfloor-(-1)^n)}{3n^2} \\[5pt]
	\cdot &
	\frac{\lfloor\frac{n}{2}\rfloor}{n} &
	\frac{(n-1)(2n-1+3(-1)^n)}{12n^2} &
	\frac{(n-1)(2n-1+3(-1)^n)}{12n^2} \\[5pt]
	\cdot &
	\cdot &
	\frac{\lfloor\frac{n}{2}\rfloor}{n} &
	\frac{(n-1)(2n-1+3(-1)^n)}{12n^2}  \\[5pt]
	\cdot & \cdot & \cdot &
	\frac{\lfloor\frac{n}{2}\rfloor}{n}  \\[5pt]
\frac{r(4n-2r-1-(-1)^n)}{4n^2} & \frac{r(4n-2r-1-(-1)^n)}{4n^2} & \frac{r(2r-1+(-1)^n)}{4n^2} & \frac{r(2r-1+(-1)^n)}{4n^2} \\
\end{array}\right)
$$
where the missing entries are determined by $p(1,1|i,j)=p(1,1|j,i)$ for $i,j\le 4$.

\begin{thm}\label{t:Snr_self_test}
The strategy $\cS_{n,r}$ is self-tested by its correlation for all $n,r\in\N$ with $r<n$.
\end{thm}

\begin{proof}
By Theorem \ref{t:Sn_self_test}, the strategy $\cS_n$ is self-tested by its correlation. Note that the projection $\fQ^{(n,r)}$ lies in the image of the span of $\{\fP_i^{(n)}\}_{i=1}^4$ under the map $\sgn$.
Therefore $\cS_{n,r}$ is self-tested by its correlation by Proposition \ref{p:chen}.
\end{proof}

\begin{cor}\label{c:st5}
Every local binary projective measurement appears in a 5-input 2-output strategy that is self-tested by its correlation.
\end{cor}

\begin{proof}
Every binary PVM is, up to unitary basis change, determined by its dimension and ranks of its projections. Therefore it suffices to consider measurements $(\fQ^{(n,r)},I-\fQ^{(n,r)})$, and these appear in the 5-input 2-output strategies $\cS_{n,r}$, self-tested by Theorem \ref{t:Snr_self_test}.
\end{proof}

Finally, we generalize Theorem \ref{t:Snr_self_test} to arbitrary $K$-PVMs. Given $r_1,\dots,r_K,n\in\N$ with $n=r_1+\cdots+r_K$, Remark \ref{r:Qalt} shows that 
$$\fQ^{(r_1,\dots,r_K)}_a:=\fQ^{(n,r_1+\cdots+r_a)}-\fQ^{(n,r_1+\cdots+r_{a-1})}$$
is a projection of rank $r_a$ for every $a=1,\dots,K$, and
$$\left(\fQ^{(r_1,\dots,r_K)}_a\right)_{a=1}^K$$
is a $K$-PVM. To it we assign a certain bipartite strategy with a mixed number of inputs and outputs.

\begin{defn}\label{d:pvm}
Let $r_1,\dots,r_K,n\in\N$ with $n=r_1+\cdots+r_K$. We define a bipartite strategy $\cS_{r_1,\dots,r_K}$ that has 4 inputs with 2 outputs and 1 input with $K$ outputs for the first party, and 4 inputs with 2 outputs for the second party:
$$\cS_{r_1,\dots,r_K}=\left(\ket{\phi_n};
\left(\fP^{(n)}_i,I-\fP^{(n)}_i\right)_{i=1}^4,
\left(\fQ^{(r_1,\dots,r_K)}_a\right)_{a=1}^K; 
\left(\fP^{(n)}_i,I-\fP^{(n)}_i\right)_{i=1}^4
\right).$$
\end{defn}

As for the correlation of $\cS_{n,r}$ from Definition \ref{d:5in}, one can derive similar (yet more involved) formulae for the correlation of $\cS_{r_1,\dots,r_K}$ using Remark \ref{r:Qalt}, and Propositions \ref{p:spectrum} and \ref{p:34to12}.

\begin{cor}\label{c:pvm}
Let $r_1,\dots,r_K,n\in\N$ with $n=r_1+\cdots+r_K$ be arbitrary. Then the strategy $\cS_{r_1,\dots,r_K}$ is self-tested by its correlation.

In particular, every single local $K$-PVM appears in a self-tested strategy that has 8 inputs with 2 outputs and 1 input with $K$ outputs.
\end{cor}

\begin{proof}
Let
$$\cS=\left(\ket{\psi};
\left(P_i,I-P_i\right)_{i=1}^4,
\left(R_a\right)_{a=1}^K; 
\left(Q_i,I-Q_i\right)_{i=1}^4
\right)$$
be a bipartite strategy with the same correlation as $\cS_{r_1,\dots,r_K}$.
Define bipartite strategies that have $3+K$ inputs with 2 outputs for the first party, and 4 inputs with 2 outputs for the second party:
\begin{align*}
\wt\cS&=\left(\ket{\phi_n};
\left(\fP^{(n)}_i,I-\fP^{(n)}_i\right)_{i=1}^4,
\left(
\fQ^{(n,r_1+\cdots+r_a)}_a,I-\fQ^{(n,r_1+\cdots+r_a)}_a
\right)_{a=1}^{K-1};\right. \\
&\qquad\qquad\qquad\qquad\qquad\qquad\quad
\left.\left(\fP^{(n)}_i,I-\fP^{(n)}_i\right)_{i=1}^4
\right),\\
\cS'&=\left(\ket{\psi};
\left(P_i,I-P_i\right)_{i=1}^4,
\left(R_1+\cdots+R_a,I-(R_1+\cdots+R_a)\right)_{a=1}^{K-1}; \right.\\
&\qquad\qquad\qquad\qquad\quad\ \quad
\left.\left(Q_i,I-Q_i\right)_{i=1}^4
\right).
\end{align*}
Since the projections $\fQ^{(n,r_1+\cdots+r_a)}$ lie in the image of the span of $\{\fP_i^{(n)}\}_{i=1}^4$ under the map $\sgn$ by Proposition \ref{p:distinct}, and the strategy $\cS_n$ is self-tested by Theorem \ref{t:Sn_self_test}, the strategy $\wt\cS$ is self-tested by a repeated application of Proposition \ref{p:chen}. Therefore $\wt\cS$ is a local dilation of $\cS'$. The same local isometries and the ancillary state show that $\cS_{r_1,\dots,r_K}$ is a local dilation of $\cS$.
\end{proof}

\section{Obstructions to constant-sized self-tests}\label{s:obs}

In a sense, maximally entangled states of all dimensions and single binary projective measurements of all dimensions and ranks can be self-tested with a constant number of inputs and outputs because they form discrete families of objects (i.e., they are parameterized by one and two natural parameters, respectively). On the other hand, there are no constant-sized self-tests for all entangled states, nor for all pairs of binary projective measurements, as implied by the results of this section (for self-tests with varying numbers of inputs, see \cite{col19} and \cite{CMV}). 
The local dimension of subsystems in a quantum strategy is not directly responsible for the absence of constant-sized self-tests; rather, dimensions of parameter spaces describing states and pairs of binary projective measurements are the obstructions to existence of uniform self-tests. The proofs of statements in this section rely on notions from real algebraic geometry \cite{BCR}.

By the singular value decomposition, every bipartite $\ket{\psi}\in\C^n\otimes\C^n$ is, up to a left-right unitary basis change, equal to
$$\sum_{i=1}^nc_i \ket{i}\!\ket{i}$$
for $c_i\ge0$ and $\sum_{i=1}^nc_i^2=1$. The numbers $c_i$ are the \emph{Schmidt coefficients} of $\ket{\psi}$. For example, all the Schmidt coefficients of $\ket{\phi_n}$ are $\frac{1}{\sqrt{n}}$. Note that $\ket{\psi}$ has full Schmidt rank if and only if $c_i>0$ for all $i$. 

\begin{prop}\label{p:nogo_states}
Let $L,K,N\in\N$ satisfy
$$L>(N(K-1)+1)^2.$$
Then for all $d_1,\dots,d_L\in\N$ there exists a bipartite state with $L$ distinct Schmidt coefficients of multiplicities $d_1,\dots,d_L$ that cannot be self-tested by $N$-inputs and $K$-outputs.
\end{prop}

\begin{proof}
Let $\bA$ denote the set of all $N$-input $K$-output bipartite quantum strategies whose states are of the form
\begin{equation}\label{e:schmidt_coef}
\ket{\psi}=\sum_{\ell=1}^L\lambda_\ell \sum_{i=d_{\ell-1}+1}^{d_\ell}\ket{i}\!\ket{i},
\qquad \lambda_1<\cdots<\lambda_L
\end{equation}
where $d_0:=0$. In particular, the states in strategies from $\bA$ have full Schmidt rank and $L$ distinct Schmidt coefficients of multiplicities $d_1,\dots,d_L$. Consider the action of $G:=\uc{d_1}\times\cdots\times \uc{d_L}$ on $\bA$, given by
$$U\cdot \big(\ket{\psi};(\cM_i)_i;(\cN_j)_j\big)=
\Big(U\otimes U\ket{\psi};(U\cM_i U^*)_i;(U\cN_j U^*)_j\Big)$$
for $U=\oplus_{\ell=1}^L U_\ell\in G$. Note that $G$ encodes precisely all actions of local unitaries that preserve the form \eqref{e:schmidt_coef} of states in strategies from $\cS$. Let $\bB$ be the quotient of $\bA$ with respect to the action of $G$, and let $\pi:\bA\to \bB$ be the canonical projection.
Given $\cS\in \bA$ let $f(\cS)\in \R^{d_1+\cdots+d_L}\otimes\R^{d_1+\cdots+d_L}$ be its state (i.e., $f$ is the projection onto the first component of the strategy).
To $\cS=(\ket{\psi};(\cM_i)_i;(\cN_j)_j)$ we also assign a tuple 
$g(\cS)\in \R^{(N(K-1)+1)^2-1}$ consisting of
\begin{align*}
&\bra{\psi}\cM_{i,a}\otimes \cN_{j,b}\ket{\psi},\qquad i,j=1,\dots,N,\ a,b=1,\dots,K-1,\\
&\bra{\psi}\cM_{i,a}\otimes I\ket{\psi},\qquad i=1,\dots,N,\ a=1,\dots,K-1,\\
&\bra{\psi}I\otimes \cN_{j,b}\ket{\psi},\qquad j=1,\dots,N,\ b=1,\dots,K-1.
\end{align*}
Note that $g(\cS)$ determines the correlation of $\cS$. The set $\bA$ is semialgebraic and the maps $f,g$ are semialgebraic \cite[Section 2]{BCR}.
Furthermore, $\bB$ is semialgebraic by \cite[Proposition 2.2.4]{BCR} since $G$ is a semialgebraic group. The maps $f,g$ factor through $\pi$, in the sense that there are semialgebraic maps $f',g'$ on $\bB$ satisfying $f'\circ \pi=f$ and $g'\circ \pi=g$.
Let $\bC\subseteq \bB$ be the set of equivalence classes $[\cS]$ such that $g'^{-1}(\{g'([\cS])\})=\{[\cS]\}$.
Then $\bC$ is also semialgebraic by \cite[Proposition 2.2.4]{BCR}.
Note that if $\cS\in \bA$ is self-tested by its correlation then $\pi(\cS)\in \bC$.
Observe that $\dim f'(\bB)=L-1$, and $\dim \bC=\dim g'(\bC)\le (N(K-1)+1)^2-1$ by \cite[Theorem 2.8.8]{BCR} since $g'|_\bC$ is injective. Surjectivity of $f'|_\bC$ would imply $\dim \bC\ge L-1$, contradicting $L-1>(N(K-1)+1)^2-1$. Therefore $f'|_\bC$ is not surjective. In particular, there exists a state $\ket{\psi}$ of the form \eqref{e:schmidt_coef} such that $\pi(\cS)\notin \bC$ for every $\cS\in f^{-1}(\{\ket{\psi}\})$. In particular, no $N$-input $K$-output strategy containing $\ket{\psi}$ is self-tested by its correlation.
\end{proof}

By the renowned theorem of Halmos \cite{halmos69}, a pair of projections $P_1,P_2\in \mtxc{n}$ is, up ot a unitary basis change, equal to
\begin{equation}\label{e:halmos}
\begin{split}
P_1&=
\ve_1 \oplus\cdots\oplus \ve_o\oplus\begin{pmatrix}
1& 0\\0&0
\end{pmatrix}\oplus\cdots\oplus \begin{pmatrix}
1& 0\\0&0
\end{pmatrix},\\
P_2&=
\ve'_1 \oplus\cdots\oplus \ve'_o\oplus\begin{pmatrix}
	\tfrac{1+\cos\alpha_1}{2}& \tfrac{\sin\alpha_1}{2}\\
	\tfrac{\sin\alpha_1}{2}&\tfrac{1-\cos\alpha_1}{2}
\end{pmatrix}\oplus\cdots\oplus
\begin{pmatrix}
	\tfrac{1+\cos\alpha_L}{2}& \tfrac{\sin\alpha_L}{2}\\
\tfrac{\sin\alpha_L}{2}&\tfrac{1-\cos\alpha_L}{2}
\end{pmatrix},
\end{split}
\end{equation}
where $\ve_i,\ve'_i\in\{0,1\}$ and $\alpha_\ell\in(0,\frac{\pi}{2})$.
The number of distinct $2\times 2$ blocks in \eqref{e:halmos} equals the number of distinct positive eigenvalues of $i(P_1P_2-P_2P_1)$.

\begin{prop}\label{p:nogo_pairs}
Let $L,N\in\N$ satisfy $L+1>(N+1)^2$.
Then for all $d_0,d_1,\dots,d_L\in\N$ there exists a pair of binary projective measurements $(P_1,I-P_1),(P_2,I-P_2)$ with $L$ distinct $2\times 2$ blocks in \eqref{e:halmos} with multiplicities $d_1,\dots,d_L$ and $d_0$ $1\times 1$ blocks, that cannot be self-tested by $N$-inputs and $2$-outputs.
\end{prop}

\begin{proof}
We proceed analogously as in the proof of Proposition \ref{p:nogo_states}. The set $\bA$ consists of $N$-input $2$-output strategies whose first two measurements are given by projections of the form \eqref{e:halmos} with $L$ angles $\alpha_\ell$ of multiplicities $d_1,\dots,d_L$. Let $f:\bA\to \mtxr{d_0+2(d_1+\cdots+d_L)}^2$ be the projection onto the pair of projections defining the first two measurements in a strategy. The group $G$ consists of all unitaries preserving the structure of \eqref{e:halmos}. Then $g,\bB,\bC$ are defined similarly as in the proof of Proposition \ref{p:nogo_states}, and the same dimension arguments apply.
\end{proof}

\begin{appendices}

\section{Distinguished projections in low dimensions}\label{app}

As a demonstration of Remark \ref{r:recursive}, we construct $\fP^{(n)}_1,\dots \fP^{(n)}_4$ for $n\le 6$.
\\ $n=1$:  $(1), (0), (0), (0)$
\\ $n=2$:  
$$\begin{pmatrix}0&0\\0&0\end{pmatrix},
\begin{pmatrix} 1 & 0 \\ 0 & 0\end{pmatrix},
\begin{pmatrix}
 \frac{1}{4} & \frac{-\sqrt{3}}{4} \\
 \frac{-\sqrt{3}}{4} & \frac{3}{4} \\
\end{pmatrix},
\begin{pmatrix}
 \frac{1}{4} & \frac{\sqrt{3}}{4} \\
 \frac{\sqrt{3}}{4} & \frac{3}{4} \\
\end{pmatrix}$$
\\ $n=3$:  
$$
\begin{pmatrix}
 1 & 0 & 0 \\
 0 & 1 & 0 \\
 0 & 0 & 0
\end{pmatrix},
\begin{pmatrix}
0 & 0 & 0 \\
 0 & \frac{4}{9} & \frac{-2 \sqrt{5}}{9} \\
 0 & \frac{-2 \sqrt{5}}{9} & \frac{5}{9}
 \end{pmatrix},
\begin{pmatrix}
 \frac{1}{3} & \frac{1}{3 \sqrt{3}} & \frac{\sqrt{5}}{3\sqrt{3}} \\
 \frac{1}{3 \sqrt{3}} & \frac{1}{9} & \frac{\sqrt{5}}{9} \\
 \frac{\sqrt{5}}{3\sqrt{3}} & \frac{\sqrt{5}}{9} & \frac{5}{9}
 \end{pmatrix},
\begin{pmatrix}
 \frac{1}{3} & \frac{-1}{3 \sqrt{3}} & \frac{-\sqrt{5}}{3\sqrt{3}} \\
 \frac{-1}{3 \sqrt{3}} & \frac{1}{9} & \frac{\sqrt{5}}{9} \\
 \frac{-\sqrt{5}}{3\sqrt{3}} & \frac{\sqrt{5}}{9} & \frac{5}{9}
 \end{pmatrix}
$$
\\ $n=4$:  
\begin{align*}
&\begin{pmatrix}
 1 & 0 & 0 & 0 \\
 0 & 0 & 0 & 0 \\
 0 & 0 & 0 & 0 \\
 0 & 0 & 0 & 0 
\end{pmatrix},
\begin{pmatrix}
 \frac{1}{4} & 0 & \frac{-\sqrt{3}}{4} & 0 \\
 0 & 1 & 0 & 0 \\
 \frac{-\sqrt{3}}{4} & 0 & \frac{3}{4} & 0 \\
 0 & 0 & 0 & 0 
\end{pmatrix},\\
&\begin{pmatrix}
 \frac{1}{4} & \frac{-\sqrt{15}}{16} & \frac{\sqrt{3}}{8} & \frac{\sqrt{21}}{16} \\
 \frac{-\sqrt{15}}{16} & \frac{3}{8} & \frac{-3 \sqrt{5}}{16} & 0 \\
 \frac{\sqrt{3}}{8} & \frac{-3 \sqrt{5}}{16} & \frac{1}{2} & \frac{-\sqrt{7}}{16} \\
 \frac{\sqrt{21}}{16} & 0 & \frac{-\sqrt{7}}{16} & \frac{7}{8}
\end{pmatrix},
\begin{pmatrix}
 \frac{1}{4} & \frac{\sqrt{15}}{16} & \frac{\sqrt{3}}{8} & \frac{-\sqrt{21}}{16} \\
 \frac{\sqrt{15}}{16} & \frac{3}{8} & \frac{3 \sqrt{5}}{16} & 0 \\
 \frac{\sqrt{3}}{8} & \frac{3 \sqrt{5}}{16} & \frac{1}{2} & \frac{\sqrt{7}}{16} \\
 \frac{-\sqrt{21}}{16} & 0 & \frac{\sqrt{7}}{16} & \frac{7}{8}\end{pmatrix}
\end{align*}
\\ $n=5$:  
\begin{align*}
&\begin{pmatrix}
 1 & 0 & 0 & 0 & 0 \\
 0 & 1 & 0 & 0 & 0 \\
 0 & 0 & 1 & 0 & 0 \\
 0 & 0 & 0 & 0 & 0 \\
 0 & 0 & 0 & 0 & 0
\end{pmatrix},
\begin{pmatrix}
 0 & 0 & 0 & 0 & 0 \\
 0 & \frac{4}{25} & 0 & \frac{-2 \sqrt{21}}{25} & 0 \\
 0 & 0 & \frac{16}{25} & 0 & \frac{-12}{25} \\
 0 & \frac{-2 \sqrt{21}}{25} & 0 & \frac{21}{25} & 0 \\
 0 & 0 & \frac{-12}{25} & 0 & \frac{9}{25}
\end{pmatrix},\\
&\begin{pmatrix}
 \frac{2}{5} & \frac{3}{5 \sqrt{5}} & 0 & \frac{\sqrt{21}}{5\sqrt{5}} & 0 \\
 \frac{3}{5 \sqrt{5}} & \frac{8}{25} & \frac{\sqrt{7}}{25} & \frac{\sqrt{21}}{25} & \frac{3 \sqrt{7}}{25} \\
 0 & \frac{\sqrt{7}}{25} & \frac{2}{25} & \frac{-\sqrt{3}}{25} & \frac{6}{25} \\
 \frac{\sqrt{21}}{5\sqrt{5}} & \frac{\sqrt{21}}{25} & \frac{-\sqrt{3}}{25} & \frac{12}{25} & \frac{-3 \sqrt{3}}{25} \\
 0 & \frac{3 \sqrt{7}}{25} & \frac{6}{25} & \frac{-3 \sqrt{3}}{25}
 & \frac{18}{25}
 \end{pmatrix}, 
\begin{pmatrix}
 \frac{2}{5} & \frac{-3}{5 \sqrt{5}} & 0 & \frac{-\sqrt{21}}{5\sqrt{5}} & 0 \\
 \frac{-3}{5 \sqrt{5}} & \frac{8}{25} & \frac{-\sqrt{7}}{25} & \frac{\sqrt{21}}{25} & \frac{-3 \sqrt{7}}{25} \\
 0 & \frac{-\sqrt{7}}{25} & \frac{2}{25} & \frac{\sqrt{3}}{25} & \frac{6}{25} \\
 \frac{-\sqrt{21}}{5\sqrt{5}} & \frac{\sqrt{21}}{25} & \frac{\sqrt{3}}{25} & \frac{12}{25} & \frac{3 \sqrt{3}}{25} \\
 0 & \frac{-3 \sqrt{7}}{25} & \frac{6}{25} & \frac{3 \sqrt{3}}{25} & \frac{18}{25}
 \end{pmatrix}
\end{align*}
\\$n=6$:
\begin{align*}
&\begin{pmatrix}
	1 & 0 & 0 & 0 & 0 & 0 \\
	0 & 1 & 0 & 0 & 0 & 0 \\
	0 & 0 & 0 & 0 & 0 & 0 \\
	0 & 0 & 0 & 0 & 0 & 0 \\
	0 & 0 & 0 & 0 & 0 & 0 \\
	0 & 0 & 0 & 0 & 0 & 0 
\end{pmatrix},
\begin{pmatrix}
	\frac{1}{9} & 0 & 0 & \frac{-2 \sqrt{2}}{9} & 0 & 0 \\
	0 & \frac{4}{9} & 0 & 0 & \frac{-2 \sqrt{5}}{9} & 0 \\
	0 & 0 & 1 & 0 & 0 & 0 \\
	\frac{-2 \sqrt{2}}{9} & 0 & 0 & \frac{8}{9} & 0 & 0 \\
	0 & \frac{-2 \sqrt{5}}{9} & 0 & 0 & \frac{5}{9} & 0 \\
	0 & 0 & 0 & 0 & 0 & 0 \end{pmatrix}, \\
&\begin{pmatrix}
	\frac{13}{36} & \frac{1}{4 \sqrt{3}} & \frac{-\sqrt{\frac{35}{3}}}{12} & \frac{\sqrt{2}}{9} & \frac{\sqrt{\frac{5}{3}}}{4} & 0 \\
	\frac{1}{4 \sqrt{3}} & \frac{7}{36} & 0 & \frac{-1}{4 \sqrt{6}} & \frac{\sqrt{5}}{9} & \frac{\sqrt{\frac{55}{6}}}{12} \\
	\frac{-\sqrt{\frac{35}{3}}}{12} & 0 & \frac{5}{12} & \frac{-\sqrt{\frac{35}{6}}}{6} & 0 & 0 \\
	\frac{\sqrt{2}}{9} & \frac{-1}{4 \sqrt{6}} & \frac{-\sqrt{\frac{35}{6}}}{6} & \frac{17}{36} & \frac{-\sqrt{\frac{5}{6}}}{4} & 0 \\
	\frac{\sqrt{\frac{5}{3}}}{4} & \frac{\sqrt{5}}{9} & 0 & \frac{-\sqrt{\frac{5}{6}}}{4} & \frac{23}{36} & \frac{-\sqrt{\frac{11}{6}}}{12} \\
	0 & \frac{\sqrt{\frac{55}{6}}}{12} & 0 & 0 & \frac{-\sqrt{\frac{11}{6}}}{12} & \frac{11}{12}\end{pmatrix},
\begin{pmatrix}
	\frac{13}{36} & \frac{-1}{4 \sqrt{3}} & \frac{\sqrt{\frac{35}{3}}}{12} & \frac{\sqrt{2}}{9} & \frac{-\sqrt{\frac{5}{3}}}{4} & 0 \\
	\frac{-1}{4 \sqrt{3}} & \frac{7}{36} & 0 & \frac{1}{4 \sqrt{6}} & \frac{\sqrt{5}}{9} & \frac{-\sqrt{\frac{55}{6}}}{12} \\
	\frac{\sqrt{\frac{35}{3}}}{12} & 0 & \frac{5}{12} & \frac{\sqrt{\frac{35}{6}}}{6} & 0 & 0 \\
	\frac{\sqrt{2}}{9} & \frac{1}{4 \sqrt{6}} & \frac{\sqrt{\frac{35}{6}}}{6} & \frac{17}{36} & \frac{\sqrt{\frac{5}{6}}}{4} & 0 \\
	\frac{-\sqrt{\frac{5}{3}}}{4} & \frac{\sqrt{5}}{9} & 0 & \frac{\sqrt{\frac{5}{6}}}{4} & \frac{23}{36} & \frac{\sqrt{\frac{11}{6}}}{12} \\
	0 & \frac{-\sqrt{\frac{55}{6}}}{12} & 0 & 0 & \frac{\sqrt{\frac{11}{6}}}{12} & \frac{11}{12}
\end{pmatrix}
\end{align*}

To obtain $\fQ_{n,r}$, one computes $\fQ_{n,r}=\sum_{i=1}^r \ket{e_i}\!\bra{e_i}$ where $\ket{e_i}$ are unit eigenvectors of $\fP^{(n)}_3+\fP^{(n)}_4$ corresponding to the $r$ smallest eigenvalues in increasing order. Examples for $r< n\le 5$ are given below.
\\ $n=2$, $r=1$:
$$\begin{pmatrix}1&0 \\ 0& 0\end{pmatrix}$$
\\ $n=3$, $r=1,2$:
$$\begin{pmatrix}
 0 & 0 & 0 \\
 0 & \frac{5}{6} & \frac{-\sqrt{5}}{6} \\
 0 & \frac{-\sqrt{5}}{6} & \frac{1}{6}
\end{pmatrix},
\begin{pmatrix}
 1 & 0 & 0 \\
 0 & \frac{5}{6} & \frac{-\sqrt{5}}{6} \\
 0 & \frac{-\sqrt{5}}{6} & \frac{1}{6}
\end{pmatrix}$$
\\ $n=4$, $r=1,2,3$:
$$\begin{pmatrix}
 \frac{3}{4} & 0 & \frac{-\sqrt{3}}{4} & 0 \\
 0 & 0 & 0 & 0 \\
 \frac{-\sqrt{3}}{4} & 0 & \frac{1}{4} & 0 \\
 0 & 0 & 0 & 0
\end{pmatrix},
\begin{pmatrix}
 \frac{3}{4} & 0 & \frac{-\sqrt{3}}{4} & 0 \\
 0 & 1 & 0 & 0 \\
 \frac{-\sqrt{3}}{4} & 0 & \frac{1}{4} & 0 \\
 0 & 0 & 0 & 0
\end{pmatrix},
\begin{pmatrix}
 1 & 0 & 0 & 0 \\
 0 & 1 & 0 & 0 \\
 0 & 0 & 1 & 0 \\
 0 & 0 & 0 & 0
\end{pmatrix}$$
\\ $n=5$, $r=1,2,3,4$:
\begin{align*}
&\begin{pmatrix}
 0 & 0 & 0 & 0 & 0 \\
 0 & 0 & 0 & 0 & 0 \\
 0 & 0 & \frac{9}{10} & 0 & \frac{-3}{10} \\
 0 & 0 & 0 & 0 & 0 \\
 0 & 0 & \frac{-3}{10} & 0 & \frac{1}{10} 
\end{pmatrix},
\begin{pmatrix}
 0 & 0 & 0 & 0 & 0 \\
 0 & \frac{7}{10} & 0 & \frac{-\sqrt{21}}{10} & 0 \\
 0 & 0 & \frac{9}{10} & 0 & \frac{-3}{10} \\
 0 & \frac{-\sqrt{21}}{10} & 0 & \frac{3}{10} & 0 \\
 0 & 0 & \frac{-3}{10} & 0 & \frac{1}{10}
\end{pmatrix},\\
&\begin{pmatrix}
 1 & 0 & 0 & 0 & 0 \\
 0 & \frac{7}{10} & 0 & \frac{-\sqrt{21}}{10} & 0 \\
 0 & 0 & \frac{9}{10} & 0 & \frac{-3}{10} \\
 0 & \frac{-\sqrt{21}}{10} & 0 & \frac{3}{10} & 0 \\
 0 & 0 & \frac{-3}{10} & 0 & \frac{1}{10}
\end{pmatrix},
\begin{pmatrix}
 1 & 0 & 0 & 0 & 0 \\
 0 & 1 & 0 & 0 & 0 \\
 0 & 0 & \frac{9}{10} & 0 & \frac{-3}{10} \\
 0 & 0 & 0 & 1 & 0 \\
 0 & 0 & \frac{-3}{10} & 0 & \frac{1}{10}
\end{pmatrix}
\end{align*}
\\ $n=6$, $r=1,2,3,4,5$:
\begin{align*}
&\begin{pmatrix}
	0 & 0 & 0 & 0 & 0 & 0 \\
	0 & \frac{5}{6} & 0 & 0 & \frac{-\sqrt{5}}{6} & 0 \\
	0 & 0 & 0 & 0 & 0 & 0 \\
	0 & 0 & 0 & 0 & 0 & 0 \\
	0 & \frac{-\sqrt{5}}{6} & 0 & 0 & \frac{1}{6} & 0 \\
	0 & 0 & 0 & 0 & 0 & 0
\end{pmatrix},
\begin{pmatrix}
	\frac{2}{3} & 0 & 0 & \frac{-\sqrt{2}}{3} & 0 & 0 \\
	0 & \frac{5}{6} & 0 & 0 & \frac{-\sqrt{5}}{6} & 0 \\
	0 & 0 & 0 & 0 & 0 & 0 \\
	\frac{-\sqrt{2}}{3} & 0 & 0 & \frac{1}{3} & 0 & 0 \\
	0 & \frac{-\sqrt{5}}{6} & 0 & 0 & \frac{1}{6} & 0 \\
	0 & 0 & 0 & 0 & 0 & 0
\end{pmatrix},\\
&\begin{pmatrix}
	\frac{2}{3} & 0 & 0 & \frac{-\sqrt{2}}{3} & 0 & 0 \\
	0 & \frac{5}{6} & 0 & 0 & \frac{-\sqrt{5}}{6} & 0 \\
	0 & 0 & 1 & 0 & 0 & 0 \\
	\frac{-\sqrt{2}}{3} & 0 & 0 & \frac{1}{3} & 0 & 0 \\
	0 & \frac{-\sqrt{5}}{6} & 0 & 0 & \frac{1}{6} & 0 \\
	0 & 0 & 0 & 0 & 0 & 0
\end{pmatrix},
\begin{pmatrix}
	1 & 0 & 0 & 0 & 0 & 0 \\
	0 & \frac{5}{6} & 0 & 0 & \frac{-\sqrt{5}}{6} & 0 \\
	0 & 0 & 1 & 0 & 0 & 0 \\
	0 & 0 & 0 & 1 & 0 & 0 \\
	0 & \frac{-\sqrt{5}}{6} & 0 & 0 & \frac{1}{6} & 0 \\
	0 & 0 & 0 & 0 & 0 & 0
\end{pmatrix},
\begin{pmatrix}
	1 & 0 & 0 & 0 & 0 & 0 \\
	0 & 1 & 0 & 0 & 0 & 0 \\
	0 & 0 & 1 & 0 & 0 & 0 \\
	0 & 0 & 0 & 1 & 0 & 0 \\
	0 & 0 & 0 & 0 & 1 & 0 \\
	0 & 0 & 0 & 0 & 0 & 0
\end{pmatrix}
\end{align*}

\end{appendices}

\bibliographystyle{abbrv}
\bibliography{selftest}

\end{document}